\documentclass[pdflatex,sn-mathphys-num]{sn-jnl}

\usepackage{graphicx}%
\usepackage{multirow}%
\usepackage{amsmath,amssymb,amsfonts}%
\usepackage{amsthm}%
\usepackage{mathrsfs}%
\usepackage[title]{appendix}%
\usepackage{xcolor}%
\usepackage{textcomp}%
\usepackage{manyfoot}%
\usepackage{booktabs}%

\usepackage{dcolumn}
\usepackage{bm}
\usepackage{enumitem}
\usepackage{tikz-cd}

\theoremstyle{thmstyleone}%
\newtheorem{theorem}{Theorem}
\newtheorem{proposition}{Proposition}  
\newtheorem{lemma}{Lemma}

\theoremstyle{thmstyletwo}%
\newtheorem{remark}{Remark}

\theoremstyle{thmstylethree}%
\newtheorem{definition}{Definition}

\begin{document}

\title[Invariant measure of deviation from Petrov type D]{An invariant
  measure of deviation from Petrov type D at the level of initial
  data}

%% Authors
%% \author*[1]{\fnm{Edgar}
%%   \sur{Gasper\'in}}\email{e.gasperin@nucleares.unam.mx}
%% \author[2]{\fnm{Jarrod L.}
%%   \sur{Williams}}\email{jrrodwilliams@gmail.com}

\author*[1]{\fnm{Edgar} \sur{Gasper\'in}}
\author[2]{\fnm{Jarrod L.} \sur{Williams}}

\affil*[1]{\orgname{Instituto de Ciencias Nucleares, Universidad
    Nacional Aut\'onoma de M\'exico}, \orgaddress{\city{Cd. Mx.},
    \postcode{04510}, \country{M\'exico}}}

\affil[2]{\orgname{Independent Researcher}, \orgaddress{\city{London},
    \country{UK}}}

\begingroup
\renewcommand\thefootnote{}
\footnotetext{*Email: e.gasperin@nucleares.unam.mx}
\footnotetext{\phantom{*}Email: jrrodwilliams@gmail.com}
\endgroup

\abstract{ In this article we describe a simple covariant
  characterisation of initial data sets which give rise to Petrov type
  D vacuum spacetime developments. As an application, we derive an
  integral invariant which, when restricted to the appropriate class
  of asymptotically Euclidean initial data sets, vanishes if and only
  if the initial dataset is isometric to initial data for the Kerr
  spacetime. As such, the invariant can be considered a measure of
  \emph{non-Kerrness} on such initial data sets. In contrast with
  other similar invariants constructed through the notion of
  ``approximate Killing spinors'', the present invariant is
  \emph{algebraic} in the sense that it is algorithmically computable
  directly from initial data without having to solve any PDEs on the
  initial data hypersurface.  }


\maketitle

%%%%%%%%%%%%%%%%%%%%%%
\section{Introduction}
%%%%%%%%%%%%%%%%%%%%%%
The Petrov classification \cite{Pet54} is an algebraic classification
of the Weyl tensor, $C_{abcd}$, based on the number of \emph{Principal
Null Directions} (\emph{PND}s). A PND is a null vector $k^a$
satisfying the condition
\begin{eqnarray}
k_{[a}C_{b]cd[e}k_{f]}k^ck^d=0,\label{Eq:PND}
\end{eqnarray}
---see \cite{GriPod09, SteMacHer80}. Although there are different ways
of presenting Petrov's classification, it is particularly transparent
when expressed in spinor notation.  The Weyl spinor can be written as
\begin{eqnarray}\label{eq:Petrovpnds}
\Psi_{ABCD}= \alpha_{(A}\beta_{B}\gamma_{D}\delta_{D)},
\end{eqnarray}
where each valence-1 spinor in equation \eqref{eq:Petrovpnds}
corresponds to a PND ---see \cite{PenRin86}.  Depending on whether
there are four distinct, two repeated, two pairs of repeated, three
repeated or four repeated PNDs, the Weyl spinor is said to be of
Petrov type I, II, D, III or N, respectively. The sixth case called
type O is the conformally flat case in which $\Psi_{ABCD}=0$. A
spacetime is said to be \emph{algebraically general} if it is of
Petrov type I and \emph{algebraically special} otherwise (cases II, D,
III, N, O).  The degree of specialisation can be visualised in the
following Penrose--Petrov diagram \cite{Ste91} where the arrows
indicate degeneration of one type into another.
\begin{figure}[h]
  \centering \includegraphics[width=0.5\textwidth]{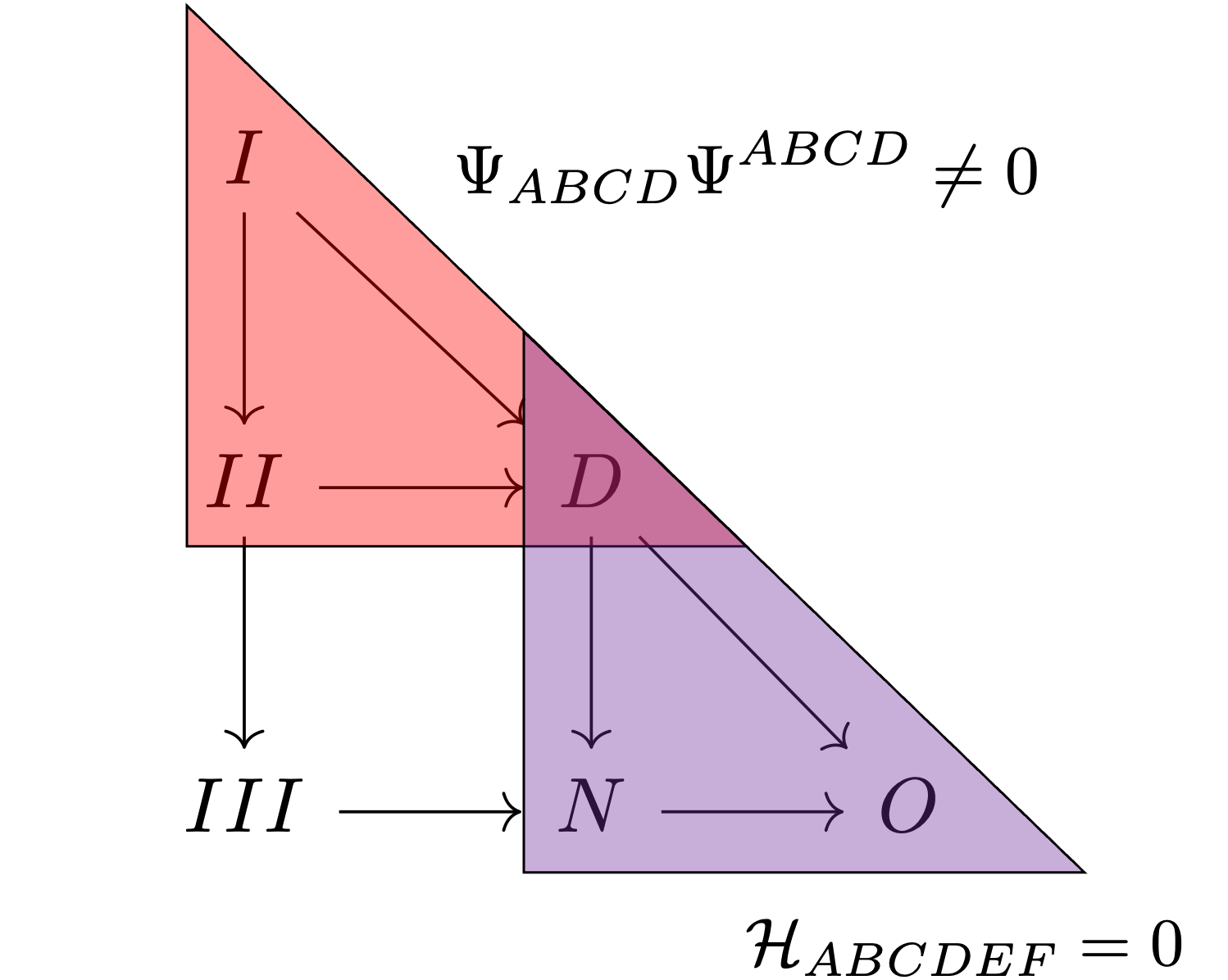}
  \caption{Penrose--Petrov diagram. Here, the Petrov types in the blue
    region are characterised by $\mathcal{H}_{ABCDEF}=0$ (see equation
    \eqref{PetrovDZeroQuantity}) and the Petrov types in the red
    region are characterised by $\Psi_{ABCD}\Psi^{ABCD}\neq 0$. At the
    intersection is Petrov type D.}
  \label{fig:example}
\end{figure}
A common technique for finding exact solutions to the Einstein field
equations is to make the simplifying assumption that the spacetime
admits a null congruence associated to a repeated PND
\cite{SteMacHer80}. Hence, many known explicit solutions to the
Einstein field equations are algebraically special.  The case of
Petrov type D is particularly important because it is the class that
contains all of the well-known \emph{explicit} solutions describing
black hole spacetimes: Schwarzschild, Reissner-Nordstr\"{o}m, Kerr and
their generalisations.  The Kerr spacetime ---see \cite{AndBacBlu16}
for a review--- is the prototypical example of a rotating black hole
solution, and is central to several open problems such as the
\emph{final state conjecture} and the \emph{black hole stability
problem} \cite{Lan21, DafRod10}.  Roughly speaking, the Kerr spacetime
is singled out of all vacuum type D solutions by the property that it
is asymptotically flat and admits a Killing spinor with a real-valued
associated Killing vector \cite{BacVal10a, BacVal10b}.

\newpage

Since many of the outstanding problems in Mathematical General
Relativity are formulated in the framework of the Cauchy problem, it
is of considerable interest to be able to characterise type D
solutions, and in particular, the Kerr solution, at the level of
initial data. A characterisation of initial data giving rise to a
Petrov type D development was given in \cite{Gar16}, see Theorem 6,
forming the basis of a characterisation of Kerr initial data therein,
and generalised to a local non-negative invariant characterisation in
\cite{Gar15}.  These characterisations, while being algorithmic, are
algebraically complicated. On the other hand, a \emph{global} approach
to quantifying \emph{non-Kerrness} was given in
\cite{BacVal10a}. However, it has the drawback that it is defined in
terms of the solution of an elliptic PDE system defined over the
initial hypersurface, which, although linear, nevertheless poses a
challenge to compute in practice. In this article, we present an
alternative characterisation of initial data for type D spacetimes,
and a resulting invariant measure of non-Kerness similar in spirit to
that of \cite{BacVal10a} but defined entirely in terms of curvature
invariants. As a result, this invariant is computable directly from
the initial data, without having to solve a PDE system on the initial
hypersurface.

\medskip

This paper is structured as follows: in Section \ref{Sec:Background},
we collect together relevant background on Petrov type, Killing
spinors and their interconnections; in Section
\ref{Sec:Characterisation} we give our characterisation of initial
data for type D spacetimes; in Section \ref{sec:ConstructInvariant} we
encode the latter characterisation in terms of a non-negative integral
invariant; finally, in Section \ref{Sec:Deviation} we give an
application of the invariant as a measure of \emph{non-Kerrness} on a
suitable class of initial data.

\medskip

\noindent Many of the calculations in this paper were carried out using the xAct
computer algebra suite, \cite{Xact}.

%%%%%%%%%%%%%%%%%%%%%%%%%
\section{Background}
\label{Sec:Background}
%%%%%%%%%%%%%%%%%%%%%%%%%

In this section we collect together the relevant background on Petrov
type, Killing spinors and their interconnections.

\subsection{Notation and conventions}
\label{Sec:NotationAndConventions}

For spinors we will follow the conventions of \cite{PenRin84}; in
particular, the metric signature is taken to be $(+,-,-,-)$.  For
spacetime tensor indices, lowercase letters from the first half of the
alphabet will be used. For spatial tensor indices, letters starting
from $i$ will be used. For spinor indices, uppercase letters will be
used. The spin metric and its inverse will be denoted by
$\epsilon_{AB}, \epsilon^{AB}$.  We will restrict here to vacuum
spacetimes; the only non-trivial curvature component is therefore the
Weyl spinor, denoted $\Psi_{ABCD}$.  In a slight abuse of notation
when writing the spinorial counterparts of tensors such as
$\xi_{AA'}=\sigma^{a}_{AA'}\xi_a$, where $\sigma^{a}_{AA'}$ are the
Infeld-Van-der-Waerden symbols, the $\sigma^{a}_{AA'}$ will be omitted
for conciseness and we will simply write $\xi_a=\xi_{AA'}$.
Occasionally, we will use index-free notation when the index structure
of an expression is obvious.
Additionally, we will make use of the so-called space-spinor
formalism ---see \cite{Val16}. For a self-contained discussion, the
basics of the formalism used in this article are described here.

\medskip

Given a timelike vector $N^a$, normsalised as $N_a N^a=1$ we consider
the spinor $N^{AA'}=N^a$, satisfying
$2N_{AA'}N^{BA'}=\epsilon_{A}{}^{B}$.
In these normalisation
conventions, a spacetime spinor $\xi_{AA'}$ splits as
\[
\xi_{AA'}=\xi N_{AA'}-\sqrt{2} N^B{}_{A'}{\xi}_{(AB)}.
\]
where $ \xi:=N^{AA'}\xi_{AA'}$ and
$\xi_{(AB)}:=\sqrt{2}N_{(A}{}^{A'}\xi_{B)A'}$.
Consequently, the Levi-Civita connection splits as
\begin{eqnarray*}
\nabla_{AA'} = N_{AA'}\mathcal{P} -\sqrt{2}
N^B{}_{A'}\mathcal{D}_{AB},
\end{eqnarray*}
in terms of the normal derivative $\mathcal{P} = N^{AA'}\nabla_{AA'}$,
and the \emph{Sen} derivative, $\mathcal{D}_{AB} =
\sqrt{2}N_{(A}{}^{A'}\nabla_{B)A'}$.  The \emph{Weingarten spinor} is
defined as
\[\chi_{ABCD} := \sqrt{2} N_D{}^{C'}\mathcal{D}_{AB}N_{CC'}.\]
Similarly, one introduces the acceleration
\[A_{AB}:=2 N_B{}^{A'}\mathcal{P} N_{AA'}.\]
If $\chi_{(A}{}^Q{}_{B)Q}=0$ then the distribution induced by
$N^{AA'}$ is integrable and $\chi_{ABCD}$ corresponds to the spinorial
counterpart of the extrinsic curvature. We will assume this to be the
case from this point onwards. To fix normalisation factors when
translating to tensorial expressions, it is enough to recall that
$\nabla_a N_b = N_{a}a_{b} + K_{ab}$ where $a^b$ is the acceleration
and $K_{ab}$ the extrinsic curvature, and observe that the above
definitions imply
\begin{eqnarray*}
\nabla_{AA'}N_{CC'}=-A_{CB}N^{B}{}_{C'}N_{AA'} +
2\chi_{ABCD}N^{B}{}_{A'}N^{D}{}_{C'}.
\end{eqnarray*}
In particular, notice that $a_a=-\frac{1}{\sqrt{2}}A_{AA'}$.
Furthermore, one introduces the operators $D_{AB}$ and $D_{N}$ via
\begin{eqnarray}  
&& D_{AB}\xi_C= \mathcal{D}_{AB}\xi_C -\sqrt{2}
  \chi_{AB}{}^Q{}_{C}\xi_Q, \\ && D_{N} \xi_A
  =\mathcal{P}\xi_{A}-\tfrac{1}{2}A_{A}{}^{B}\xi_{B}, \label{eq:DNSpinor}
\end{eqnarray}
extending their definition to spinors of higher valence analogously.
On one hand, $D_{AB}$ corresponds to the space-spinor counterpart of
the intrinsic Levi-Civita connection on the $3-$manifold $\mathcal{S}$
with normal vector $N^a$ as embedded in $\mathcal{M}$. On the other
hand, the action of $D_{N}$ is given by
\begin{equation}
\label{eq:DNForTensors}
D_N \xi_i = h_{i}{}^aN^b\nabla_b \xi_a.
\end{equation}
The relation between $\mathcal{P}$ and $D_N$, when restricted to act
on spatial vectors, is given by
\begin{align}\label{eq:PtoDNVectors}
\mathcal{P}\xi_{i} = D_{N}\xi_{i} + i \epsilon_{ijk}a^j\xi^k,
\end{align}
and is extended to tensors again via the Leibniz rule. A more detailed
discussion of these operations in terms of space-spinors is given in
Appendix \ref{Ap:NormDer}.  Though we will not need this here, we note
that these operators can be extended to act on spacetime spinors such
that they
satisfy \[D_{AB}\epsilon_{CD}=D_{AB}N_{AA'}=D_{N}\epsilon_{AB}=D_{N}
N_{AA'}=0.\]

\noindent The space-spinor conjugate of $\hat{\bm \xi}$ of any spinor
$\bm \xi$ is constructed by taking its complex conjugate and
transvecting with $\bm N$. For instance, the space-spinor conjugate of
the Weyl spinor $\Psi_{ABCD}$ is given by
\[
 \widehat{\Psi}_{ABCD}=N_{A}{}^{A'}N_{B}{}^{B'}N_{C}{}^{C'}N_{D}{}^{D'}\bar{\Psi}_{A'B'C'D'}.
 \]
See \cite{Val16} for further details on the space-spinor formalism.
For any valence-$n$ spinor we define
\begin{eqnarray*}
\Vert \mathcal{Q}\Vert^2 := \mathcal{Q}_{A_{1}A_{2}...A_{n}}\widehat{\mathcal{Q}}^{A_{1}A_{2}...A_{n}} \geq 0.
\end{eqnarray*}
For even $n=2m$, this agrees with the norm computed on the tensorial
counterpart\footnote{Due to the signature $(+,-,-,-)$, there is a
factor of $-1$ inherited from raising the indices with the
negative-definite spatial (inverse) metric $h^{ij}$; the $(-1)^m$
factor compensates to give a positive-definite norm.}  $\Vert \bm
Q\Vert^2 := (-1)^m Q_{i_1\cdots i_m}\bar{Q}^{i_1\cdots i_m}$.  As is
convention, we denote by $\lbrace \bm o, \bm \iota\rbrace$ a
\emph{spin dyad}; that is to say, a pair of valence-$1$ spinors
satisfying $o_A\iota^A=1$. It is also convenient to write
\[\epsilon_{\bm 0}{}^A = o^A,  \quad \epsilon_{\bm 1}{}^A = \iota^A, \qquad \epsilon^{\bm 0}{}_A = -\iota_A,  \quad \epsilon^{\bm 1}{}_A =
o_A,\] with index raising and lowering performed with respect to
$\epsilon^{AB}$ and $\epsilon_{AB}$.  In terms of the above, we define
the following \emph{spin coefficients} for the Sen connection:
\begin{equation}
\gamma_{\bm A \bm B}{}^{\bm C}{}_{\bm D} := -\epsilon_{\bm
  D}{}^Q\epsilon_{\bm A}{}^A\epsilon_{\bm
  B}{}^B\mathcal{D}_{AB}\epsilon^{\bm C}{}_Q,
\end{equation}
which encode a combination of the connection coefficients of $D_{AB}$ and the
extrinsic curvature.
Since it will be needed later, we give here the
the spinorial counterpart of the 3-dimensional volume form $\epsilon_{ijk}$ on $\mathcal{S}$:
\begin{align}\label{3dimVolume}
\epsilon_{ABCDEF}=\frac{i}{\sqrt{2}}(\epsilon_{AC}\epsilon_{BE}\epsilon_{DF}+\epsilon_{BD}\epsilon_{AF}\epsilon_{CE}).  
\end{align}

\subsection{Initial data sets and the Weyl spinor}
\label{sec:IDsetsWeylSpinor}

 An initial dataset for a vacuum spacetime (with vanishing
 cosmological constant) is defined as a triple $(\mathcal{S}, h_{ij},
 K_{ij})$, $\mathcal{S}$ being a $3-$manifold with Riemannian metric
 $h_{ij}$ and $K_{ij}$ a symmetric tensor, the extrinsic curvature,
 satisfying the vacuum Einstein constraint equations:
%\begin{subequations}
    \begin{eqnarray}
        && r[\bm h] - K_{ij}K^{ij} + K^2 = 0,\\ && D^iK_{ij} - D_j K =
      0.
    \end{eqnarray}
%\end{subequations}
Here, $r[\bm h]$ denotes the Ricci scalar curvature of $h_{ij}$, and
$K=K_i{}^i$. A solution describes the initial data for a vacuum
spacetime $(\mathcal{M},\bm g)$, with $h_{ij}, K_{ij}$ corresponding
to the first and second fundamental forms of the embedding
$\mathcal{S}\hookrightarrow\mathcal{M}$.  There is a vast literature
on existence and uniqueness results for the Cauchy problem in General
Relativity \cite{Ren05} and although it might be possible to reduce
the regularity requirements of the initial data, from now on it will
be assumed that the initial data is smooth so that we can apply the
basic local existence theorems of \cite{ChoYor80} to ensure smoothness
of the solution. Observe that other characterisation results of the
Kerr spacetime such as \cite{Mar99, Mar00, BacVal10a, BacVal10b}
implicitly work in the smooth category as it is based on the Killing
spinor initial data result of \cite{GarVal08c} ---see Remark
\ref{remark:RegularityAssumptions}.  When discussing the spacetime
development of the initial data, $\mathcal{D}^{+}(\mathcal{S})$ will
denote the future domain of dependence of $\mathcal{S}$.

\medskip

\noindent
\noindent
The Einstein constraints are the trace parts of the Gauss--Codazzi--Mainardi equations:
%%\begin{subequations}
\begin{align}
r_{ij} - E_{ij}\big|_{\mathcal{S}} - K_{i}{}^k K_{jk} + K K_{ij} &= 0, \label{GCM1} \\
\epsilon_{i}{}^{kl} D_{k} K_{lj} - B_{ij}\big|_{\mathcal{S}} &= 0. \label{GCM2}
\end{align}
%%\end{subequations}
Here, $E_{ij}\big|_{\mathcal{S}}$ and $B_{ij}\big|_{\mathcal{S}}$ are the
pullbacks to $\mathcal{S} \hookrightarrow \mathcal{M}$ of the
\emph{electric} and \emph{magnetic} parts of the Weyl tensor, defined by
\begin{equation*}  
E_{ab} = C_{acbd} N^c N^d, \qquad
B_{ab} = C^*_{acbd} N^c N^d.
\end{equation*}
with $N^a$ the unit normal to the hypersurface and
$C^*_{abcd}=-\frac{1}{2}\epsilon_{cd}{}^{fg}C_{abfg}$. 
The Weyl curvature is determined fully by $E_{ab}, B_{ab}$ as follows 
\begin{eqnarray}\label{eq_Weyl_To_EB_tensor}
    C_{abcd} = 2E_{b[c} g_{d]a} - 2E_{a[c} g_{d]b} +
    2\epsilon_{cdef} B_{[a}{}^{f} N_{b]} N^{e} + 2\epsilon_{abef}
      B_{[c}{}^{f} N_{d]} N^{e}
\end{eqnarray}
---see \cite{Val16}, for example, for further details.  Note that
$E_{ab}$ and $B_{ab}$ are intrinsic to $\mathcal{S}$ in the sense that
$N^aE_{ab}=N^aB_{ab}=0$.  Hence, when considering a spacetime
foliation $\mathcal{S} _{t} \subset \mathcal{M}$ for which, in some
local coordinates $(t,x^k)$, the hypersurface $\mathcal{S}$
corresponds to the $t=0$ slice, one has $E_{ij}=E_{ij}(t,x^k)$ and
$B_{ij}=B_{ij}(t,x^k)$ while $E_{ij}|_{\mathcal{S}}=E_{ij}(0,x^k)$ and
$B_{ij}|_{\mathcal{S}}=B_{ij}(0,x^k)$. Although introducing the symbol
$\quad|_{\mathcal{S}}$ in the notation may seem unnecessary, we do it
to emphasise that a given quantity is directly computable from initial
data $(\mathcal{S}, h_{ij},K_{ij})$. For example, in this case,
through equations \eqref{GCM1}--\eqref{GCM2}.

\medskip

\noindent The Weyl tensor is ``spinorialised" as follows
\begin{align}\label{eq_WeyTensor_to_spinor}
C_{abcd} = \Psi_{ABCD}\bar{\epsilon}_{A'B'}\bar{\epsilon}_{C'D'} +
\bar{\Psi}_{A'B'C'D'}\epsilon_{AB}\epsilon_{CD},
\end{align}
where $\Psi_{ABCD}=\Psi_{(ABCD)}$ is the \emph{Weyl spinor}.  The
spinorial counterpart of equation \eqref{eq_Weyl_To_EB_tensor} is
given by
\begin{eqnarray*} 
\Psi_{ABCD} = E_{ABCD} + iB_{ABCD},
\end{eqnarray*}
with $E_{ABCD}$, $B_{ABCD}$ denoting the space-spinorial counterparts
of $E_{ij}$ and $B_{ij}$, which can be recovered directly from
$\Psi_{ABCD}$ as follows
\begin{flalign}
  E_{ABCD}:&= \frac{1}{2}(\Psi_{ABCD} + \widehat{\Psi}_{ABCD}),
  \\ B_{ABCD}:&= \frac{i}{2}(-\Psi_{ABCD} +\widehat{\Psi}_{ABCD}).
\end{flalign}
Alternatively, one can introduce a complex tensor given by
$\Psi_{ij}=E_{ij} + i B_{ij}$ that succinctly encodes the geometric
information of $C_{abcd}$.

\subsection{A covariant spacetime characterisation of Petrov type D spacetimes}

As usual in the discussion of the Petrov classification,
one considers the following $\mathbb{C}$-valued scalars
%\begin{subequations}
\begin{eqnarray}
&& I := \Psi_{ij}\Psi^{ij}\equiv \Psi_{ABCD}\Psi^{ABCD}, \\
&& J := \Psi_i{}^j\Psi_j{}^k\Psi_k{}^i\equiv \Psi_{AB}{}^{CD}\Psi_{CD}{}^{EF}\Psi_{EF}{}^{AB}.\qquad 
\end{eqnarray}
%\end{subequations}
A spacetime is \emph{algebraically special} if 
\[I^3 -6J^2=0,\] 
and, moreover, of Petrov type III, N or O if $I=J=0$
---see \cite{SteMacHer80}.  On the other hand, a spacetime is Petrov type D
if there exists a spin
dyad $\lbrace\bm{o},\bm\iota\rbrace$ in terms of which
\begin{equation}
    \Psi_{ABCD} = \Psi o_{(A}o_B\iota_C\iota_{D)},\label{eq:WeylSpinorTypeDdyad}
\end{equation} 
for some non-zero complex-valued scalar function $\Psi$. We call such a dyad
\emph{Petrov-adapted}. Note that there is no unique choice of dyad \cite{PenRin86}; 
it is determined only up to 
spin boosts and dyad exchange symmetry:
%\begin{subequations}
\begin{eqnarray} 
&& o_A\rightarrow e^{i\phi}o_A, \quad \iota_A \rightarrow e^{-i\phi}\iota_A ,\label{SpinBoosts} \\
&& o_A\rightarrow\iota_A,\quad \iota_A\rightarrow -o_A.\label{Interchange}
\end{eqnarray}
%\end{subequations}

\noindent For type D, we have $I=\Psi^2/6$ and $J=-\Psi^3/36$, which,
of course, trivially satisfy the algebraically special condition.  In
particular, notice that the condition $\Psi \neq 0$ (and hence $I \neq 0$)
holds everywhere on Kerr ---see e.g. \cite{Mar00}.  The following
concomitant of the Weyl spinor is central to the forthcoming
discussion:
 \begin{equation}
 \mathcal{H}_{ABCDEF}:=\Psi_{PQR(A}\Psi^{QR}{}_{BC}\Psi^{P}{}_{DEF)}.
 \label{PetrovDZeroQuantity}
 \end{equation}
The relevance of this object is the content of the next Lemma:
 \begin{lemma} (Penrose \& Rindler, \cite{PenRin86})
    $(\mathcal{M},\bm g)$ is type D or more special at $p\in
   \mathcal{M}$ if and only if $\mathcal{H}_{ABCDEF}\vert_p=0$.
\end{lemma}
\noindent See pg. 80 of \cite{Ste91} or equation (8.6.3) of \cite{PenRin86} and
the discussion there for a detailed proof.

\medskip
 Note also that $\mathcal{H}_{ABCDEF}\vert_{\mathcal{S}}$ is
 computable from the initial data and it is natural to call
 \emph{type-D initial data}, the data for which
 $\mathcal{H}_{ABCDEF}\vert_{\mathcal{S}}=0$.  Nonetheless care is
 needed with the language here since the expression `type-D initial
 data' can be potentially misleading as
 $\mathcal{H}_{ABCDEF}\vert_{\mathcal{S}}=0$ is a necessary but
 \emph{insufficient} condition to guarantee that the spacetime
 development will be Petrov type D. To describe initial data sets
 whose development is guaranteed to be of Petrov type D, we introduce
 the term \emph{propagating-type-D initial data}.

\begin{remark}
It can be shown that the condition
\begin{equation}
    \mathcal{H}_{ABCD}:=J\Psi_{ABCD} - I \Psi_{(AB}{}^{PQ}\Psi_{CD)PQ}
    = 0
\end{equation}
characterises the property of being strictly more special than type
$II$.  Combining this with $I\neq 0$ gives a second characterisation
of type D.  In fact, much of the forthcoming analysis can be carried
out, with only minor adjustments, with $\mathcal{H}_{ABCD}$ in place
of $\mathcal{H}_{ABCDEF}$. However, we have chosen to use the
six-index object as it only involves up to cubic terms in the Weyl
spinor.
\end{remark}

The existence of hidden symmetries (encoded by Killing spinors) is
closely related to the Petrov type, as discussed in the remainder of
this subsection.  A \emph{Killing spinor} is a symmetric $2-$spinor,
$\kappa_{AB}$, satisfying the equation
\begin{align}\label{eq_Killing_spinor_eq}
  \nabla_{A'(A} \kappa_{BC)} = 0.
\end{align}
It is straightforward to show that on a vacuum spacetime, given a
Killing spinor, $\xi_{AA'} := \nabla^B{}_{A'}\kappa_{AB}$ is a
(complex-valued, in general) Killing vector. Moreover, one can show
that $\kappa_{AB}$ must necessarily satisfy the integrability
condition
\begin{align}\label{eq_buchdahl}
  \Psi_{(ABC}{}^{Q}\kappa_{D)Q}=0,
\end{align}
this being called the \emph{Buchdahl constraint}. This constrains the
spacetime to be of Petrov type D, N, or O ---see \cite{GarVal08c}.
Furthermore, in particular, on a region of spacetime which is type D,
one can construct a valence$-2$ Killing spinor in terms of a
Petrov-adapted spin dyad $\lbrace o, \iota\rbrace$ as follows:
\begin{eqnarray} \label{KSAnsatz}
    \kappa_{AB} = \Psi^{-1/3} o_{(A}\iota_{B)}
\end{eqnarray}
---see \cite{Pen70}.  The fact that this expression satisfies the
Killing spinor equation is guaranteed by the second Bianchi identity,
which in spinorial formulation reads
\begin{equation}
\nabla^A{}_{A'}\Psi_{ABCD}=0, \label{Bianchi}
\end{equation}
in vacuum ---see \cite{PenRin84}. Conversely, given a Killing spinor
on some open spacetime region $\mathcal{V}$, it follows that
$\Psi_{ABCD} \propto \kappa_{(AB}\kappa_{CD)}$ at each
$p\in\mathcal{V}$ by virtue of the Buchdahl constraint, equation
\eqref{eq_buchdahl}. Consequently, if $\kappa_{AB}$ is algebraically
general ($\kappa_{AB}\kappa^{AB}\neq 0$) at some point $p$, then the
spacetime is necessarily of Petrov type D at $p$.

\subsection{Killing spinor initial data}

Given the close connection between Killing spinors and Petrov type, it
is of interest to be able to encode the existence of a Killing spinor
at the level of initial data, that is to say \emph{Killing spinor
initial data}. This can be thought of as a spinorial analogue of the
Killing Initial Data (KID) equations, \cite{BeigChr96}.

\medskip

The Killing spinor initial data equations were first given in
\cite{GarVal08c} and further streamlined in \cite{BacVal10a}. In the
latter, it is shown that if $\varkappa_{AB}=\varkappa_{(AB)}$
satisfies
%\begin{subequations}
    \begin{eqnarray}
      && \mathcal{D }_{(AB}\varkappa_{CD)} =
      0, \label{SpatialKSEq}\\ && \varkappa_{(A}{}^Q\Psi_{BCD)Q} = 0
      ,\label{BuchdahlConstraint}
    \end{eqnarray}
%\end{subequations}
on an open set $\mathcal{U}\subset\mathcal{S}$ and additionally
satisfies $\varkappa_{AB}\varkappa^{AB} \neq 0$, then it constitutes
initial data for a Killing spinor for a vacuum spacetime; indeed a
Killing spinor $\kappa_{AB}$ can be constructed as the solution of the
following initial value problem
\begin{equation}\label{KillingSpinorIVP}
    \left\{
\begin{array}{ll}
	 \square \kappa_{AB} - \Psi_{ABCD}\kappa^{CD} = 0 & \qquad
         \text{on}~\mathcal{D}^{+}(\mathcal{U}),\\ \kappa_{AB} =
         \varkappa_{AB} &\qquad\text{on}~\mathcal{U},\\ \mathcal{P}
         \kappa_{AB} =- \mathcal{D}_{(A}{}^Q\varkappa_{B)Q}
         &\qquad\text{on}~\mathcal{U}.
\end{array} \right.
\end{equation}

The above facts suggest the following approach to characterising type
D initial data. First, verify that the Weyl curvature is of type D on
the initial hypersurface (necessary condition). Then define
\begin{equation} \label{KSAnsatzIntrinsic}
    \varkappa_{AB} := \Psi^{-1/3} o_{(A}\iota_{B)},
\end{equation}
in terms of a Petrov-adapted spin dyad, and find \emph{supplementary}
conditions under which $\varkappa_{AB}$ solves the Killing spinor
initial data equations,
\eqref{SpatialKSEq}--\eqref{BuchdahlConstraint} on $\mathcal{U}$.
Then, use the Killing spinor $\kappa_{AB}$ resulting from solving the
initial value problem \eqref{KillingSpinorIVP} to constrain the Petrov
type of the ambient spacetime development, thereby \emph{propagating}
the Petrov type off the initial hypersurface.  In other words, the
supplementary conditions are the conditions needed to upgrade
\emph{type-D initial data} to \emph{propagating-type-D initial data}
(necessary and sufficient conditions).
\begin{remark}\label{remark:RegularityAssumptions}
    Note that the proof given in \cite{GarVal08c, BacVal10b} of the
    existence of a Killing spinor, $\kappa_{AB}$, as a solution to
    \eqref{KillingSpinorIVP}, assumes a smooth spacetime and a smooth
    Killing spinor candidate, $\varkappa_{AB}$. It would be of
    interest to extend this result to low-regularity spacetimes and
    low-regularity initial data $\varkappa_{AB}$. However, this is
    beyond the scope of this paper.
\end{remark}

%%%%%%%%%%%%%%%%%%%%%%%%%%%%%%%%%%%%%%%%%%%%%%%%%%%%%%%%%%
\section{Characterising propagating-type-D initial data}
\label{Sec:Characterisation}
%%%%%%%%%%%%%%%%%%%%%%%%%%%%%%%%%%%%%%%%%%%%%%%%%%%%%%%%%%

In this section we derive two equivalent
characterisations of propagating-type-D initial data: one given
in terms of a Petrov-adapted dyad and one manifestly covariant.
We begin with the 1+3 split of the Bianchi
identity \eqref{Bianchi} with respect to the spacetime
foliation, which reads
%\begin{subequations}
\begin{eqnarray}
&& \mathcal{P} \Psi_{ABCD} -
  \sqrt{2}\mathcal{D}_{(A}{}^Q\Psi_{BCD)Q} =
  0, \label{NormalDerivOfWeyl}\\ &&
  \mathcal{D}^{AB}\Psi_{ABCD}=0. \label{GaussConstraint}
\end{eqnarray}
%\end{subequations} 
We will refer to \eqref{GaussConstraint} as the \emph{Gauss constraint}. For
notational convenience, let us define $\dot{\Psi}_{ABCD} := \mathcal{P} \Psi_{ABCD}$. Notice that the equation,
\[ \dot{\Psi}_{ABCD}\equiv \sqrt{2}\mathcal{D}_{(A}{}^Q\Psi_{BCD)Q},\]
which follows from equation
\eqref{NormalDerivOfWeyl}, is manifestly intrinsic to the
hypersurface and therefore $(\bm\Psi|_{\mathcal{S}},\dot{\bm\Psi}|_{\mathcal{S}})$ is computable from the initial data $(\bm h, \bm K)$.
\\

\noindent It is clear that if $\mathcal{H}_{ABCDEF}=0$ on
$\mathcal{D}^{+}(\mathcal{U})$, then \emph{necessarily} one must have
\begin{equation}\dot{\mathcal{H}}_{ABCDEF}:=\mathcal{P}\mathcal{H}_{ABCDEF}=0
  \quad \text{on}\quad \mathcal{U}.
\end{equation}
Moreover, this condition can be recast as a manifestly intrinsic
condition by virtue of equation \eqref{NormalDerivOfWeyl}: 
\begin{eqnarray}
    \dot{\mathcal{H}}_{ABCDEF} = 2\dot{\Psi}_{PQR(A}\Psi^{QR}{}_{BC}\Psi^{P}{}_{DEF)}
    +\Psi_{PQR(A}\Psi^{QR}{}_{BC}\dot{\Psi}^{P}{}_{DEF)}.  \label{PDotIntrinsic}
\end{eqnarray}
What is remarkable is that, as we shall see, the conditions
\begin{equation}
    \mathcal{H}_{ABCDEF}=\dot{\mathcal{H}}_{ABCDEF}=0 \quad \text{on}\quad
    \mathcal{U} ,\label{TypeDInitialData}
\end{equation}
are in fact \emph{sufficient} to ensure propagation of Petrov
type D, provided $I\neq 0$ on $\mathcal{U}$.
\\

The first step is to derive the supplementary conditions ensuring that
$\varkappa_{AB}$ given by equation \eqref{KSAnsatz} satisfies equation
\eqref{SpatialKSEq}. Notice that, in contrast, equation
\eqref{BuchdahlConstraint} is automatically satisfied.  The following
proposition, which gives our first (non-covariant) characterisation of
propagating-type-D data, can be thought of as a corollary of Theorem 3
of \cite{GasWill22}. In the interest of being self-contained, we spell
out the details here.
\begin{proposition}
\label{Prop:Characterisation}
    Let $\mathcal{U}$ be an open subset of an initial dataset, on
    which the curvature is of type D. Let
    $\lbrace\bm{o},\bm\iota\rbrace$ be an adapted (but otherwise
    general) spin dyad. Then there exists an open neighbourhood of the
    spacetime development, containing $\mathcal{U}$, on which the
    curvature is of Petrov type D if and only if
    %\begin{subequations}
    \begin{eqnarray}
        && \gamma_{\bm 1 \bm 1}{}^{\bm 0}{}_{\bm 1}\equiv
      \iota^A\iota^B\iota^C\mathcal{D}_{AB}\iota_C =
      0, \label{ExtremalNPLambda}\\ && \gamma_{\bm 0 \bm 0}{}^{\bm
        1}{}_{\bm 0} \equiv -o^Ao^Bo^C\mathcal{D}_{AB}o_C =
      0, \label{ExtremalNPSigma}
    \end{eqnarray}  
    %\end{subequations}
    hold on $\mathcal{U}$.
\end{proposition}
\begin{proof}
    Suppose that the initial dataset is of type D and consider $\varkappa_{AB} = \Psi^{-1/3}
    o_{(A}\iota_{B)}$, which is clearly well-defined by virtue of the assumption $I(\equiv \Psi^2/6)\neq
    0$ on $\mathcal{U}$. Note that 
    \[\varkappa_{AB}\varkappa^{AB} =
    -\tfrac{1}{2}\Psi^{-2/3} \neq 0.\]
    Since the Gauss constraint
    \eqref{GaussConstraint} is intrinsic to the hypersurface, we can
    substitute $\Psi_{ABCD}=\Psi o_{(A}o_B\iota_C\iota_{D)}$ therein 
    to get
%\begin{subequations}
\begin{eqnarray}
  &&  \mathcal{D}_{\bm 0 \bm 0}\Psi = -6 \Psi \gamma_{\bm 0 \bm 1}{}^{\bm 1}{}_{\bm 0},
  \label{GaussConstraint0}\\
  &&  \mathcal{D}_{\bm 0 \bm 1}\Psi =
  -\tfrac{3}{2}\Psi(\gamma_{\bm 0 \bm 0}{}^{\bm 0}{}_{\bm 1}
  + \gamma_{\bm 1 \bm 1}{}^{\bm 1}{}_{\bm 0}), \label{GaussConstraint1}\\
  &&  \mathcal{D}_{\bm 1 \bm 1}\Psi = -6\Psi \gamma_{\bm 0 \bm 1}{}^{\bm 0}{}_{\bm 1},
  \label{GaussConstraint2}
\end{eqnarray}
%\end{subequations}
on $\mathcal{U}$. It follows from a short computation that these
equations are equivalent to the $\bm0\bm0\bm0\bm1,~\bm0\bm0\bm1\bm1$
and $\bm0\bm1\bm1\bm1$ components of the equation
$\mathcal{D}_{(AB}\varkappa_{CD)}=0$. The remaining two
(\emph{extremal}) components of $\mathcal{D}_{(AB}\varkappa_{CD)}$ are
given by
\[ \mathcal{D}_{\bm 0 \bm 0}\varkappa_{\bm 0 \bm 0} = \Psi^{-1/3} \gamma_{\bm 0 \bm 0}{}^{\bm 1}{}_{\bm 0}, \qquad  \mathcal{D}_{\bm 1\bm 1}\varkappa_{\bm 1\bm 1} = \Psi^{-1/3} \gamma_{\bm 1 \bm 1}{}^{\bm 0}{}_{\bm 1}.\]
Hence, if equations \eqref{ExtremalNPLambda}--\eqref{ExtremalNPSigma}
are satisfied, then $\mathcal{D}_{(AB}\varkappa_{CD)}=0$. It is also
straightforward to see that $\varkappa_{AB}$ satisfies the Buchdahl
constraint. Hence, $\varkappa_{AB}$ satisfies the Killing spinor
initial data equations \eqref{SpatialKSEq}--\eqref{BuchdahlConstraint}
and therefore gives rise to a Killing spinor $\kappa_{AB}$ on the
spacetime development. By continuity, $\kappa_{AB}\kappa^{AB}\neq 0$
on a sufficiently small neighbourhood, on which the Buchdahl
constraint for $\kappa_{AB}$ implies that $\Psi_{ABCD}=\Psi^{5/3}
\kappa_{(AB}\kappa_{CD)}\neq 0$. As a result, the spacetime
development is of Petrov type D in a suitably small \emph{spacetime}
neighbourhood of $\mathcal{U}$.
\end{proof}
Note that we follow essentially the same construction as in
\cite{Pen70}, in which $\kappa_{AB}=\Psi^{-1/3}o_{(A}\iota_{B)}$ is
shown to be a Killing spinor on a type D \emph{spacetime}. The main
difference lies in the fact that in \cite{Pen70} the full Bianchi
identities are used instead of only the Gauss constraint; here, since
we only assume a priori that the curvature is type D when
\emph{restricted} to $\mathcal{S}$, we cannot assume that
$\Psi_{ABCD}=\Psi o_{(A}o_B\iota_C\iota_{D)}$ away from $\mathcal{S}$
---in particular, we cannot substitute this relation into the
evolutionary components of the Bianchi identities,
\eqref{NormalDerivOfWeyl}. This additional information is instead
contained in the supplementary conditions
\eqref{ExtremalNPLambda}--\eqref{ExtremalNPSigma}.  \\

As a sanity check, note that the supplementary conditions
\eqref{ExtremalNPLambda}--\eqref{ExtremalNPSigma} are invariant under
spin boosts \eqref{SpinBoosts} and the spin dyad exchange symmetry
\eqref{Interchange} ---that is to say, they are not dependent on the
particular choice of Petrov-adapted spin dyad. It is also instructive
to write the supplementary conditions
\eqref{ExtremalNPLambda}--\eqref{ExtremalNPSigma} in terms of the
better-known spin coefficients of the NP formalism,
\cite{PenRin84}. Accordingly, let the normal to the hypersurface
$\mathcal{S}\hookrightarrow\mathcal{M}$ be given by
\[ N_{AA'} = N_0 l_{AA'} + N_1 m_{AA'} + \bar{N}_1 \bar{m}_{AA'} + N_2 n_{AA'},\]
where \quad
$l_{AA'} = o_A\bar{o}_{A'}, \quad m_{AA'}=o_A\bar{\iota}_{A'}, \quad n_{AA'}= \iota_A\bar{\iota}_{A'}.$
The tetrad vectors $\bm l, \bm m, \bar{\bm m}, \bm n$ are PNDs and a
short computation then shows that
    \[ \gamma_{\bm 1 \bm 1}{}^{\bm 0}{}_{\bm 1} = \sqrt{2}(N_0 \lambda + N_1 \nu), \qquad \gamma_{\bm 0 \bm 0}{}^{\bm 1}{}_{\bm 0} = \sqrt{2}(N_2 \sigma + \bar{N}_1\kappa),\]
    in terms of the NP spin coefficients
\begin{align*}
    & \lambda=-\bar{m}^a\bar{m}^b\nabla_bn_a, & \nu = -\bar{m}^a
  n^b\nabla_bn_a, & \sigma=m^am^b\nabla_bl_a, &
  \kappa=m^al^b\nabla_bl_a.
\end{align*}
Hence, the conditions $\gamma_{\bm 1 \bm 1}{}^{\bm 0}{}_{\bm
  1}=\gamma_{\bm 0 \bm 0}{}^{\bm 1}{}_{\bm 0}=0$ are consistent with a
well-known consequence of the \emph{Goldberg--Sachs Theorem},
\cite{GolSac09}, namely that
\[\lambda = \nu = \sigma = \kappa=0\] 
for a Petrov type D spacetime.  Moreover, it is straightforward to see
that $\gamma_{\bm 1 \bm 1}{}^{\bm 0}{}_{\bm 1}$ and $\gamma_{\bm 0 \bm
  0}{}^{\bm 1}{}_{\bm 0}$ are, in general, the only degrees of freedom
of $\lambda, \nu, \sigma, \kappa$ that are intrinsic to the
hypersurface $\mathcal{S}$, all other combinations involving normal
derivatives of either $l^a$ or $n^a$.

\medskip

 Although it is possible, in principle, to compute $\gamma_{\bm 1 \bm
  1}{}^{\bm 0}{}_{\bm 1}$ and $\gamma_{\bm 0 \bm 0}{}^{\bm 1}{}_{\bm
  0}$ at the level of initial data, it is of course undesirable to
have to first construct the Petrov-adapted frame\footnote{In order to
so, one could project the equations \eqref{Eq:PND}, or the
\emph{Bel--Debever} conditions \cite{Bel60, Bel62, Deb59}, onto
$\mathcal{S}$ and solve the resulting intrinsic equations.}. An
alternative is given by the following Lemma, which realises the spin
connection coefficients as the components of the covariant quantity
$\dot{\mathcal{H}}_{ABCDEF}$:

\begin{lemma}
\label{Lemma:Pdot}
If the curvature is of Petrov type D on
$\mathcal{U}\subset\mathcal{S}$, then
\begin{equation} 
\dot{\mathcal{H}}_{ABCDEF} = -\frac{1}{8}\Psi^3\left(\gamma_{\bm 1 \bm
  1}{}^{\bm 0}{}_{\bm 1} o_{(A}o_Bo_Co_D o_E\iota_{F)}  +\gamma_{\bm 0 \bm 0}{}^{\bm 1}{}_{\bm
  0}o_{(A}\iota_B\iota_C\iota_D\iota_E\iota_{F)}\right)
\end{equation}
on $\mathcal{U}$, in terms of a Petrov-adapted spin dyad $\lbrace \bm
o, \bm\iota\rbrace$.
\end{lemma}
\begin{proof}
  Follows by a direct computation from equation \eqref{PDotIntrinsic},
  using relations \eqref{GaussConstraint0}--\eqref{GaussConstraint2}.
\end{proof}

Combining Lemma \ref{Lemma:Pdot} and Proposition
\ref{Prop:Characterisation}, we then obtain the following:
\begin{theorem}\label{thm:PetrovTypeD-detect}
    Let $\mathcal{U}\subset \mathcal{S}$ be an open subset of a smooth
    initial dataset $(\mathcal{S},\bm h, \bm K)$, on which $I\neq
    0$. Then there exists an open neighbourhood of the resulting
    spacetime development on which the curvature is of Petrov type D
    if and only if $\mathcal{H}_{ABCDEF} = \dot{\mathcal{H}}_{ABCDEF}
    = 0$ on $\mathcal{U}$.
\end{theorem}
\begin{proof}
    The only if direction is immediate. Conversely, suppose that
    $\mathcal{H}_{ABCDEF} = \dot{\mathcal{H}}_{ABCDEF} = 0$ on
    $\mathcal{U}$. Then the curvature is of type D on $\mathcal{U}$
    (in particular, $\Psi\neq 0$) and Lemma \ref{Lemma:Pdot} along
    with $\dot{\mathcal{H}}_{ABCDEF}=0$ imply that $\gamma_{\bm 1 \bm
      1}{}^{\bm 0}{}_{\bm 1}=\gamma_{\bm 0 \bm 0}{}^{\bm 1}{}_{\bm
      0}=0$ in a Petrov-adapted dyad. Proposition
    \ref{Prop:Characterisation} then implies that the spacetime
    development is of type D in $\mathcal{D}^{+}(\mathcal{U})$.
\end{proof}
\begin{remark}
    If one opted to use the projected normal derivative $D_{N}$, as
    given by equation \eqref{eq:DNSpinor}, instead of $\mathcal{P}$
    the result holds identically since
    \begin{equation*}
            D_{N}\mathcal{H}_{ABCDEF}= \dot{\mathcal{H}}_{ABCDEF}-3
            A_{(A}{}^{Q}\mathcal{H}_{BCDEF)Q}.
    \end{equation*}
    In other words,
    $\mathcal{H}_{ABCDEF}=\dot{\mathcal{H}}_{ABCDEF}=0$ is equivalent
    to $ \mathcal{H}_{ABCDEF} =D_{N}\mathcal{H}_{ABCDEF} =0$.
\end{remark}

\begin{remark}\label{rem:ConfGenTheorem1}
    Although the discussion given in this paper assumes the vacuum
    Einstein field equations hold, a formally identical Petrov type D
    characterisation for initial data for Friedrich's \emph{conformal
    Einstein field equations} (CEFEs) \cite{Fri81} can be trivially
    obtained. In fact, revisiting the discussion leading to Theorem
    \ref{thm:PetrovTypeD-detect} one realises that the only place
    where the vacuum Einstein field equations were used was in
    equation \eqref{Bianchi}. Noticing that the equation for the
    \emph{rescaled Weyl spinor} $\phi_{ABCD}$ is formally identical to
    equation \eqref{Bianchi} and the \emph{conformal Killing spinor
    initial data equations} of \cite{GasWill22} are formally identical
    to equations \eqref{SpatialKSEq} and \eqref{BuchdahlConstraint},
    then one concludes that Theorem \ref{thm:PetrovTypeD-detect} holds
    for initial data for the vacuum CEFEs formally replacing
    $\Psi_{ABCD}$ with $\phi_{ABCD}$ in the definition of
    $\mathcal{H}_{ABCDEF}$.
\end{remark}

%%%%%%%%%%%%%%%%%%%%%%%%%%%%%%%%%%%%%
\section{Constructing an invariant}
\label{sec:ConstructInvariant}
%%%%%%%%%%%%%%%%%%%%%%%%%%%%%%%%%%%%%

In addition to being simple to compute, the covariant
characterisation given by Theorem \ref{thm:PetrovTypeD-detect} has the added benefit that it can be used to quantify
\emph{deviation} from the property of being propagating-type-D. 
With this application in mind, it is then natural to consider 
\begin{align*}
&\mathcal{I}_1(\mathcal{U}, \bm h, \bm K) := \int_{\mathcal{U}}
  \Vert \bm{\mathcal{H}} \Vert^2 ~d\text{vol}_{\bm h}, &&
  \mathcal{I}_2(\mathcal{U}, \bm h, \bm K) := \int_{\mathcal{U}}
  \Vert\bm{\mathcal{\dot{H}}}\Vert^2 ~d\text{vol}_{\bm h},
\end{align*}
where $d\text{vol}_{\bm h}$ denotes the volume-form on $(\mathcal{S},\bm h)$,
 while $\bm{\mathcal{H}}$ and $\bm{\mathcal{\dot{H}}}$ denote $\mathcal{H}_{ABCDEF}$ and $\dot{\mathcal{H}}_{ABCDEF}$, as given in equations \eqref{PetrovDZeroQuantity}
 and \eqref{PDotIntrinsic}, respectively. 
 Notice, however, that the physical units of $\mathcal{H}_{ABCDEF}$ and $\dot{\mathcal{H}}_{ABCDEF}$ (and hence of $\mathcal{I}_1$ and $\mathcal{I}_2$) differ. Indeed,
 \begin{align*}
   [\mathcal{H}] = L^{-6} , \qquad [\dot{\mathcal{H}}] = L^{-7},
 \end{align*}
 where $L$ represents the spatial length in geometric units.
 If $\mathcal{U}$ were to have some characteristic length scale, $\ell$, one might consider $\mathcal{I}_{1} + \ell^2 \mathcal{I}_{2}$ as a measure of deviation from propagating-type-D data. 
In the absence of a geometrically motivated reference scale in $\mathcal{U}$ in general, however, it is not clear how one might combine $\mathcal{I}_1$ and $\mathcal{I}_2$ into a single invariant.

\medskip

\noindent Nonetheless, we can arrive at a single invariant if we
restrict our attention to \emph{asymptotically-Euclidean} data and if
take $\mathcal{U}=\mathcal{S}$; accordingly, our invariant will be a
global rather than a local one.  To be self-contained we recall the
following:
\begin{definition}\label{def_asympt_Euclidean}
An initial data set $(\mathcal{S},\bm h, \bm K)$ is
asymptotically-Euclidean if there exists some compact set
$\mathcal{B}$, diffeomorphic to a ball, such that
$\mathcal{S}\setminus\mathcal{B}$ is a disjoint union of open sets
$\mathcal{S}_n$, with $n \in \mathbb{N}$, which are diffeomorphic to
the complement of a closed ball in $\mathbb{R}^3$ and for each
asymptotic end $\mathcal{S}_n$ there exist (asymptotically Cartesian)
coordinates $\{x^i\}$ in which
\[ h_{ij} = -\delta_{ij} + \mathcal{O}_k(r^{-q}), \qquad K_{ij}=\mathcal{O}_{k-1}(r^{-1-q}),\]
where $r := \sqrt{(x^{1})^2+(x^{2})^2+(x^{3})^2}$, for some $k>1$ and
$0<q<1$.
\end{definition}
Here $k$ indicates denotes the fall-off rate up to $k$ derivatives,
namely $f \in \mathcal{O}_{k}(r^{-q}) \implies \partial^l f \in
\mathcal{O}(r^{-q-l})$ for $l=0,\cdots, k$ ---see \cite{BeigChr96,
  Hua10}, for example.  It follows that for such data,
\[E_{ij} = \mathcal{O}_{k-2}(r^{-2-q}), \qquad B_{ij}=\mathcal{O}_{k-2}(r^{-2-q}).\]
Notice that for data satisfying these conditions,
$\mathcal{H}_{ABCDEF}=\mathcal{O}_{k-2}(r^{-3q-6})$ so that
$\mathcal{H}_{ABCDEF}=0$ at spatial infinity.

\medskip

Instead of constructing an invariant using $\mathcal{H}_{ABCDEF}$ directly,
it is convenient to use its spatial derivatives $D_{PQ}\mathcal{H}_{ABCDEF}$ ---denoted in index-free notation as $\bm D\bm{\mathcal{H}}$. 
Then, considering
data in the asymptotically Euclidean class and following the same approach as taken in \cite{dain2004new, ValWil17} we obtain:
\begin{theorem}\label{Coro:TypeDInvariant}
Let  $(\mathcal{S},\bm h, \bm K)$ be a smooth asymptotically Euclidean initial data set,
satisfying $I\neq 0$ everywhere on $\mathcal{S}$. 
Then the invariant 
    \begin{align}
    \mathcal{I}(\mathcal{S}, \bm h, \bm K) :=  \int_{\mathcal{S}}\left( \Vert \bm D\bm{\mathcal{H}}\Vert^2 +  \Vert \bm{\mathcal{\dot{H}}}\Vert^2\right)~d\text{vol}_{\bm h} \label{InvariantDef}
\end{align}
is well-defined and vanishes if and only if  $(\mathcal{S},\bm h, \bm K)$ is propagating-type-D initial data; that is to say, if and only if the spacetime development is Petrov type D in some open neighbourhood of $\mathcal{S}$.
\end{theorem}
\begin{proof}
If the initial data is of propagating-type-D then trivially $\mathcal{I}(\mathcal{S}, \bm h, \bm K) =0$.
To see that the converse is also true, note that if $ \mathcal{I}(\mathcal{S}, \bm h, \bm K) =0$ then
$\Vert \bm D\bm{\mathcal{H}}\Vert^2=\Vert\bm{\mathcal{\dot{H}}}\Vert^2=0$.  Hence, $\bm D\bm{\mathcal{H}}=\bm{\mathcal{\dot{H}}}=0$.
On the other hand, notice that

\begin{align}
D_{PQ} (\Vert \bm{\mathcal{H}} \Vert^2) = &
\mathcal{H}^{ABCDEF}D_{PQ}\hat{\mathcal{H}}_{ABCDEF}+
\hat{\mathcal{H}}^{ABCDEF}D_{PQ}\mathcal{H}_{ABCDEF} \nonumber \\ = &
2 \mathfrak{Re}
(\hat{\mathcal{H}}^{ABCDEF}D_{PQ}\mathcal{H}_{ABCDEF}),
\label{eq:DHsquare}
\end{align}
where we have used that $D_{AB}N_{C}{}^{C'}=0$.  Now, since $\bm
D\bm{\mathcal{H}}=0$, then using equation \eqref{eq:DHsquare} one has
$D_{PQ} (\Vert \bm{\mathcal{H}}\Vert^2)=0$. Thus, $\Vert
\bm{\mathcal{H}}\Vert^2 = c$ on $\mathcal{S}$ where $c$ is constant.
Exploiting the initial data asymptotic conditions, one concludes that
$c=0$ and hence $\bm{\mathcal{H}}=0$ on $\mathcal{S}$.  Together,
these conditions read $\bm{\mathcal{H}} =\bm{\mathcal{\dot{H}}}=0 $ on
$\mathcal{S}$. The conclusion then follows from Theorem
\ref{thm:PetrovTypeD-detect}.
\end{proof}
 In the asymptotically Euclidean case, one could consider that the
 natural length scale (in geometric units) in the problem is the ADM
 mass $m_{ADM}$, and hence alternatively use as invariant the
 following quantity:
        \begin{align}
     \tilde{\mathcal{I}}(\mathcal{S}, \bm h, \bm K) :=
     \int_{\mathcal{S}}\left( \Vert \bm{\mathcal{\dot{H}}}\Vert^2
     +\frac{1}{m_{ADM}^2} \Vert \bm{\mathcal{H}}\Vert^2
     \right)~d\text{vol}_{\bm h} ,\label{InvariantAlternative}
\end{align}
for initial data with $m_{ADM} \neq 0$.

\medskip

Although the calculations presented in this paper are particularly
clean using spinor notation, we emphasise that one can express the
invariants introduced above in tensorial rather than spinorial form.
In the remainder of this section we detail how to obtain and compute
the tensor counterparts of $\mathcal{H}_{ABCDEF}$ and
$\dot{\mathcal{H}}_{ABCDEF}$.  Using equation \eqref{3dimVolume}, a
direct calculation shows that
 \begin{align}
     & \epsilon_{i}{}^{lm}\Psi_{jl}\Psi_{k}{}^{p}\Psi_{pm} =
   \epsilon_{AB}{}^{PQGH}\Psi_{CDPQ}\Psi_{EF}{}^{JK}\Psi_{JKGH}
   \nonumber \\ &
   \phantom{\epsilon_{i}{}^{lm}\Psi_{jl}\Psi_{k}{}^{p}\Psi_{pm}}=
   \frac{i}{\sqrt{2}}\Psi_{CD(A}{}^{P}\Psi_{B)PGH}\Psi_{EF}{}^{GH}.
 \end{align}
 Recalling that in space-spinor formalism a total symmetrisation
 corresponds to taking the symmetric trace-free part of a spatial
 tensor, \cite{Val16}, and observing that
 \begin{equation}
     \mathcal{H}_{ijk}:=-i\sqrt{2}\epsilon_{(i}{}^{lm}\Psi_{j|l|}\Psi_{k)}{}^{p}\Psi_{pm}
     \label{eq:HInTensors}
 \end{equation}
 is trace-free, one concludes that $\mathcal{H}_{ijk}$ is the tensor
 counterpart of $\mathcal{H}_{ABCDEF}$.  Furthermore, using that
 \[
 D_N \epsilon_{ijk} = h_{i}{}^{b}h_{j}{}^{c}h_{k}{}^{d}N^a\nabla_a
 \epsilon_{bcd} =\epsilon_{aijk} a^a =0,\] and equation
 \eqref{eq:PtoDNVectors}, one gets
 $\dot{\epsilon}_{ijk}=\mathcal{P}\epsilon_{ijk}=0$, and so
\begin{align}
    \mathcal{\dot{H}}_{ijk} = &
    -i\sqrt{2}\left[\epsilon_{(i}{}^{lm}\dot{\Psi}_{j|l|}\Psi_{k)}{}^{p}\Psi_{pm}
      \right.\nonumber\\ & +
      \epsilon_{(i}{}^{lm}\Psi_{j|l|}\dot{\Psi}_{k)}{}^{p}\Psi_{pm} +
      \left. \epsilon_{(i}{}^{lm}\Psi_{j|l|}\Psi_{k)}{}^{p}\dot{\Psi}_{pm}\right], \hspace{5mm}
    \label{eq:HDotInTensors}
\end{align} 
where, again, the dot notation is a shorthand for application of the
operator $\mathcal{P}$.  For completeness $\Psi_{ij}$ and
$\dot{\Psi}_{ij}$ are given by
  \begin{flalign}
    \Psi_{ij} &= E_{ij} + iB_{ij}, \\ \dot{\Psi}_{ij} &\equiv
    D_{N}\Psi _{ij} - 2 i a^{l} \Psi _{(i}{}^{k} \epsilon _{j)kl}
    \nonumber \\ &= i \epsilon _{kl(i} D^{k}\Psi _{j)}{}^{l} + 3 \Psi
    _{(i}{}^{k} K_{j)k} -2 K\Psi _{ij} - \Psi ^{kl} K_{kl} h _{ij}
    , \label{EvolutionaryBianchiInTensors}
 \end{flalign}
     and where the second equality in
     \eqref{EvolutionaryBianchiInTensors} follows from the second
     Bianchi identity i.e. the tensorial analogue of equation
     \eqref{NormalDerivOfWeyl}.  Clearly the invariant of
     Theorem~\ref{Coro:TypeDInvariant} is algebraically computable at
     the level of initial data, as $E_{ij}$ and $B_{ij}$ are
     expressible in terms of initial data using the
     Gauss--Codazzi--Mainardi equations \eqref{GCM1}--\eqref{GCM2}.

%%%%%%%%%%%%%%%%%%%%%%%%%%%%%%%%%%%%%%%%%%%%%%%%%%%%%%%%
\section{Quantifying deviation from Kerr initial data}
\label{Sec:Deviation}
%%%%%%%%%%%%%%%%%%%%%%%%%%%%%%%%%%%%%%%%%%%%%%%%%%%%%%%%

In the previous section we gave a characterisation of
propagating-type-D initial data sets, namely
\begin{eqnarray*}
\mathcal{H}_{ABCDEF}=\dot{\mathcal{H}}_{ABCDEF}=0, \quad I\neq 0 \qquad \text{on}~\mathcal{S}.
\end{eqnarray*}
The aim of the present section is to identify a class of initial data
for which these conditions are sufficient to single out Kerr initial
data. The class of initial data will be specified in terms of its
asymptotic properties, containing as a strict subset the \emph{boosted
asymptotically Schwarzschildean} data sets considered in similar
works, \cite{BacVal10a, BacVal10b}.

\medskip

A natural approach to singling out the Kerr spacetime would be to
eliminate those Kinnersley metrics, \cite{Kin69}, which are
incompatible with the assumed regularity and asymptotic conditions. A
similar approach was taken for instance in \cite{Val05} to
characterise the Schwarzschild spacetime exploiting \emph{Zakharov's
property}. However, as pointed out in Remark 3 of \cite{Mar99}, a
drawback of such an approach is that the derivation of the Kinnersley
metrics implicitly assumes analyticity, and is therefore overly
restrictive.  Instead, we choose to follow a similar approach to that
of \cite{BacVal10a, BacVal10b}, relying on Mars' characterisations of
the Kerr spacetime among stationary spacetimes \cite{Mar99, Mar00}.
We start by recalling the following result from \cite{BacVal10b}:
\begin{theorem}\label{Thm:JuanThomas}
(Valiente Kroon \& B\"{a}ckdahl, \cite{BacVal10b}) Let $(\mathcal{M},
  g_{ab})$ be a smooth vacuum spacetime satisfying $I \neq 0$ on
  $\mathcal{M}$. Then $(\mathcal{M}, g_{ab})$ is locally isometric to
  the Kerr spacetime if and only if the following conditions are
  satisfied:
\begin{enumerate}[label=(\roman*)]
    \item there exists a Killing spinor, $\pi_{AB}$, such that the
      associated Killing vector, $\eta^{AA'}:=\nabla^{BA'}\pi_B{}^A$,
      is real;
    \item the spacetime $(\mathcal{M}, g_{ab})$ has a stationary
      asymptotically flat 4-end with non-vanishing Komar mass in which
      $\eta^{AA'}$ tends to a time translation.
\end{enumerate}
\end{theorem}
\begin{remark}\label{remark:kid_translational}
{\em Given a KID set $(N, Y^i)$ on $\mathcal{U}\subset\mathcal{S}$, a
  Killing vector is obtained ---see \cite{BeiChr97}--- by solving the
  following initial value problem (IVP)
\begin{equation}
    \left\{
\begin{array}{ll}
	 \square X^a = 0 & \qquad
         \text{on}~\mathcal{D}^{+}(\mathcal{U}),\\ X^a =(N, Y^i)
         &\qquad\text{on}~\mathcal{U},\\ N^b\nabla_b X^a = -D^a N +
         K_{b}{}^aY^b &\qquad\text{on}~\mathcal{U}.
\end{array} \right.
\end{equation}

As shown in Appendix \ref{Sec:Boost}, an application of the theory in
\cite{ChrMur81} shows that, if $(N, Y^i) = (A^0, A^i) +
\mathcal{O}_k(r^{-q})$ on $\mathcal{U}$, where $A^0$ and $A^i$ are
constants and the initial data $(\mathcal{S},\bm h,\bm K)$ is
asymptotically Euclidean of order $(k,q)$ with $k>3$ and $q>0$, then
$X^a \rightarrow (A^0, A^i)$ as $r \rightarrow \infty$ in
$\mathcal{D}^+(\mathcal{U})$, i.e. the Killing vector remains
asymptotically translational.  }
\end{remark}

In \cite{BacVal10a, BacVal10b}, the Killing spinor initial data
equations are used to reduce this characterisation of the Kerr
spacetime to the level of an initial dataset $(\mathcal{S}, \bm h, \bm
K)$, this forming the basis of their construction of an invariant
measuring \emph{non-Kerrness}.  The authors consider a class of
initial data which they term \emph{boosted asymptotically Euclidean},
these being a special case (see section \ref{Sec:boostedAS} for the
explicit formulae) of initial data of the following form:
%\begin{subequations}
    \begin{eqnarray}
        && h_{ij} = -\left(1 + \frac{2A}{r} \right)\delta_{ij}
      \nonumber \\ && \qquad\quad -
      \frac{2\alpha}{r}\left(\frac{2x_{i}x_{j}}{r^2} - \delta_{ij}
      \right) +
      \mathcal{O}_{k}(r^{-1-q}),\qquad\qquad \label{HuangSolh}\\ &&
      K_{ij} = \frac{\beta}{r^2}\left(\frac{2x_{i}x_{j}}{r^2} -
      \delta_{ij} \right) +
      \mathcal{O}_{k-1}(r^{-2-q}).\label{HuangSolK}
    \end{eqnarray}
%\end{subequations} 
Here, $A$ is a constant and $\alpha=\alpha(\theta, \varphi)$,
$\beta=\beta(\theta, \varphi)$ are functions on $\mathbb{S}^2$.
Initial data of this more general form were first discussed in
\cite{Beig87} and later rigorously shown to exist for sufficiently
regular $\alpha, \beta$, in \cite{Hua10}.

\medskip

The approach in \cite{BacVal10a, BacVal10b} relies on solving an
elliptic PDE on $\mathcal{S}$ to construct an ``approximate Killing
spinor" and resulting approximate Killing vector.  We emphasise that,
in contrast to \cite{BacVal10a, BacVal10b}, in the forthcoming
discussion there is no analogous construct of an approximate Killing
spinor or vector, making the characterisation obtained in this paper,
in this sense, \emph{algebraic}.  This is clearly advantageous not
only from the point of view of Mathematical Relativity, but also for
numerical applications where the invariant can be monitored at each
time slice $\mathcal{S}_{t}$ in numerical evolutions of, say, compact
binaries to examine how quickly the final configuration converges to a
member of the Kerr black hole family.

\medskip

 The structure of this section is as follows. In subsection
\ref{Sec:boostedAS}, we first consider the special case of
\emph{boosted asymptotically Schwarzschildean} data sets. In doing so,
we recover a similar result (see Theorem
\ref{Thm:AsymptoticallySchwarzschildean}) to that of \cite{BacVal10a,
  BacVal10b}. In subsection \ref{Sec:more_general}, we prove a more
general result, Theorem \ref{Thm:MoreGeneral}, for a broader class of
data using asymptotic properties of Killing initial data on
asymptotically Euclidean data sets. Although the results of subsection
\ref{Sec:more_general} subsume those of \ref{Sec:boostedAS}, the
approach taken in section \ref{Sec:boostedAS} relies on less
sophisticated machinery and also makes connections with other areas of
the literature, in particular the work of Saez et al \cite{FerSae09};
it is for this reason that we have opted to include the special case
here.

\subsection{A special case: \emph{boosted
asymptotically–Schwarzschildean data}}
\label{Sec:boostedAS}

In this section, we restrict attention to \emph{boosted asymptotically
Schwarzschildean} data sets. These are initial data sets of the form
\eqref{HuangSolh}--\eqref{HuangSolK}, with
    %\begin{subequations}
    \begin{eqnarray}
         A=m/\sqrt{1-v^2},  \quad \alpha
      = \frac{2m^2+4\nu^2}{(m^2+\nu^2)^{1/2}} - 2A, \quad \beta
      =\frac{E\nu(3m^2+2\nu^2)}{(m^2+\nu^2)^{3/2}}, \label{AandbetaForSchwarzschildean}
    \end{eqnarray}
    %\end{subequations}    
where $m>0$ and $\vert v\vert < 1$ are constants, and $\nu =
-mv\cos\theta/\sqrt{1-v^2}$. These initial data sets have asymptotics
consistent with a boosted Schwarzschild black hole, with boost vector
given by $v^i\partial_i = v\partial_z$. Here, without loss of
generality we have chosen our asymptotically-Cartesian coordinate
system such that the boost vector is aligned with the $z-$axis, this
being achieved by a rotation of a generic coordinate basis. The more
general (i.e. non-coordinate adapted) form of the metric can be found
in section 6.5 of \cite{BacVal10b}. For such initial data, the ADM
$4-$momentum is given by
\begin{equation}\label{ADMForSchwarz}
    p^a = \frac{m}{\sqrt{1-v^2}}\left(1, ~v^i\right),
\end{equation}
resulting in ADM mass $m_{ADM}=m>0$.

\medskip

In order to be able to apply Theorem \ref{Thm:JuanThomas}, we need to
construct a real, timelike Killing vector. As we have seen above, for
a type D spacetime there is a canonical Killing vector
$\xi^{AA'}$. Therefore, to single out the Kerr spacetime, it would
suffice to show that $\xi^a=\xi^{AA'}$ (or some complex-constant
rescaling, thereof) is real and asymptotically timelike.  Following
Ferrando--Saez \cite{FerSae09} ---see also \cite{Gar16}--- for a type
D spacetime we define
 \[\mathcal{Q}_{abcd} := \mathcal{C}_{abcd} -\tfrac{1}{12}\Psi(g_{ac}g_{bd} - g_{ad}g_{bc} + i\epsilon_{abcd}),\]
where $\Psi$ is the only non-zero component of the Weyl spinor as in
equation \eqref{eq:WeylSpinorTypeDdyad} , $\epsilon_{abcd}$ denotes
the volume form for $g_{ab}$, and $\mathcal{C}_{abcd}$ is the
\emph{self-dual Weyl tensor},
\[\mathcal{C}_{abcd}:=\tfrac{1}{2}(C_{abcd} + iC^*_{abcd}) = \Psi_{ABCD} \epsilon_{A'B'} \epsilon_{C'D'}.\]  
In \cite{FerSae09} it is shown that for a type D spacetime, there
exists a bi-vector $\mathcal{U}_{ab}$ such that
\begin{eqnarray}
    \mathcal{Q}_{abcd} =
    \Psi\mathcal{U}_{ab}\mathcal{U}_{cd}, \label{BivectorDecompOfQ}
\end{eqnarray}
and that, whenever $\Psi\neq 0$,
\begin{eqnarray}
    \xi^b := \tfrac{3}{2}\Psi^{-1/3}\nabla_{a}\mathcal{U}^{ab}
\end{eqnarray}    
defines, in general a complex-valued Killing vector, which moreover
satisfies the identity
\begin{eqnarray}
    \Psi^{-11/3}\mathcal{Q}_{abcd}(\nabla^b \Psi)(\nabla^d \Psi) =
    \xi_a \xi_c. \label{Saez}
\end{eqnarray}
This Killing vector in fact coincides with the canonical Killing
vector $\xi^{AA'}:=\nabla^{BA'}\kappa_{B}{}^A$ ---see Appendix
\ref{Sec:Killingvectorexpression}--- which justifies our choice of
notation.
\begin{remark}
    This construction \eqref{BivectorDecompOfQ}--\eqref{Saez} forms
    the basis of the characterisations of propagating type D and Kerr
    initial data given in \cite{Gar16} (see \cite{Gar15}, also). The
    characterisations given there (c.f. Theorem 6 of \cite{Gar16}),
    require three derivatives of the Weyl curvature (equivalently,
    five derivatives of the metric). Here, our characterisation via
    the invariant $\mathcal{I}$ (c.f. Theorem
    \ref{Coro:TypeDInvariant}, above) requires only one derivative of
    the Weyl curvature (equivalently, three derivatives of the
    metric).
\end{remark}

\noindent More generally, for any asymptotically-Euclidean manifold
(not necessarily of Petrov-type D) satisfying $I\neq 0$, we can define
\begin{equation}
    Q_{ac}:= \psi^{-11/3}\tilde{\mathcal{Q}}_{abcd}(\nabla^b
    \psi)(\nabla^d \psi)
\label{Qdef}
\end{equation}
where $\psi:=-6J/I$ and
\[ \tilde{\mathcal{Q}}_{abcd} := \mathcal{C}_{abcd} -\tfrac{1}{12}\psi(g_{ac}g_{bd} - g_{ad}g_{bc} + i\epsilon_{abcd}). \]
It is clear that if the initial data is type D, then $\psi=\Psi$, in
which case it follows from equation \eqref{Saez} that
\begin{equation}
    Q_{ab}=\xi_a\xi_b. \label{Saez2}
\end{equation}
Our approach here will therefore be to directly compute the
asymptotics of the expression $Q_{ab}$ from the initial data, and to
infer the implied asymptotics of $\xi^a$ from equation
\eqref{Saez2}. We denote the leading $r^{-3}$ components of $E_{ij},
B_{ij}$ as $\mathcal{E}_{ij}, \mathcal{B}_{ij}$, so that
\[ E_{ij} = \mathcal{E}_{ij} + \mathcal{O}_1(r^{-3-q}), \quad B_{ij} = \mathcal{B}_{ij} + \mathcal{O}_1(r^{-3-q}).\]
Using equations \eqref{AsympE}--\eqref{AsympB}, we find
\begin{align}
     \mathcal{E}_{ij}dx^{i}dx^{j} &=
     -\frac{m}{r^3}\frac{(1-v^2)^{3/2}(2+v^2\sin^2\theta)}{(1-v^2\sin^2\theta)^{5/2}}
     dr^2+\frac{m}{r}\left(\frac{1-v^2}{1-v^2\sin^2\theta}\right)^{3/2}
     d\theta^2 \nonumber\\ &+
     \frac{m}{r}\frac{(1-v^2)^{3/2}(1+2v^2\sin^2\theta)\sin^2\theta}{(1-v^2\sin^2\theta)^{5/2}}d\varphi^2, \label{EForSchwarz}\\ \mathcal{B}_{ij}dx^{i}dx^{j}
     &= -\frac{6mv}{r^2}
     \frac{(1-v^2)^{3/2}\sin^2\theta}{(1-v^2\sin^\theta)^{5/2}}drd\varphi. \label{BForSchwarz}
\end{align}
Note that from equation \eqref{EvolutionaryBianchiInTensors} we have
that
\begin{equation}
    \dot{\Psi}_{ij} = i\text{rot}_2(\Psi)_{ij} +
    \mathcal{O}(r^{-4-q}), \label{PsiDotEquation}
\end{equation}
where $\text{rot}_2: \text{Sym}^2(T^*\mathcal{S})\rightarrow
\text{Sym}^2(T^*\mathcal{S})$ is defined as
\[ \text{rot}_2(\Psi)_{ij} = \mathring{\epsilon}^{kl}{}_{(i} \partial_{\vert k\vert}\Psi _{j)l},\]
with $\mathring{\epsilon}_{ijk}$ denoting the Levi--Civita tensor for
the flat metric and with index-raising performed with respect to the
inverse of the flat metric.  Together, \eqref{PsiDotEquation} and
\eqref{eq:RotB}--\eqref{eq:RotE} yield asymptotic expansions for
$\dot{\mathcal{E}}_{ij}, \dot{\mathcal{B}}_{ij}$. Combining all of the
above,
%\begin{subequations}
\begin{eqnarray}
    && I = \frac{6 m^2}{r^6} \left(\frac{1-v^2}{1 -
    v^2\sin^2\theta}\right)^3 +
  \mathcal{O}(r^{-6-q}), \label{IforAsympSchwarz}\\ && J = \frac{6
    m^3}{r^9} \left(\frac{1-v^2}{1 - v^2\sin^2\theta}\right)^{9/2} +
  \mathcal{O}(r^{-9-q}), \\ && \dot{I} =
  2(\dot{\mathcal{E}}^{ij}\mathcal{E}_{ij} -
  \dot{\mathcal{B}}^{ij}\mathcal{B}_{ij}) +
  \mathcal{O}(r^{-7-q})\nonumber\\ && \quad =\frac{36 m^2 v}{r^7}
  \frac{(1 - v^2)^3 \cos\theta}{(1 - v^2\sin^2\theta)^4 } +
  \mathcal{O}(r^{-7-q}),\\ && \dot{J} =
  3\dot{\mathcal{E}}_{i}{}^{j}\mathcal{E}_{j}{}^{k}\mathcal{E}_{k}{}^{i}
  -
  3\dot{\mathcal{E}}_{i}{}^{j}\mathcal{B}_{j}{}^{k}\mathcal{B}_{k}{}^{i}
  -
  6\mathcal{E}_{i}{}^{j}\mathcal{B}_{j}{}^{k}\dot{\mathcal{B}}_{k}{}^{i}
  + \mathcal{O}(r^{-10-q}) \nonumber \\ &&\quad =
  \frac{54m^3v}{r^{10}}\frac{(1-v^2)^{9/2}\cos\theta}{(1-v^2\sin^2\theta)^{11/2}}+
  \mathcal{O}(r^{-10-q}).\label{JdotforAsympSchwarz}
\end{eqnarray}
%\end{subequations}
Note that $I, J, \dot{I}, \dot{J}$ are all real-valued to leading
order.
\begin{remark}
Observe that the algebraically special condition $I^3-6J^2=0$ holds
asymptotically to first order; indeed, the stronger condition
\[\mathcal{H}_{ijk}=\mathcal{O}(r^{-9-q})\]
can be shown to hold ---that is to say that the leading-order
($r^{-9}$) term of $\mathcal{H}_{ijk}$ vanishes--- consistent with the
initial data being asymptotically type D.
\end{remark}
\noindent 
Equations \eqref{IforAsympSchwarz}--\eqref{JdotforAsympSchwarz} then
give
%\begin{subequations}
\begin{eqnarray}
     && \psi = -6J/I = -\frac{6 m}{r^3}
  \left(\frac{1-v^2}{1 - v^2\sin^2\theta}\right)^{3/2} +
  \mathcal{O}(r^{-3-q}), \label{AsymptoticPsi}\\ && \dot{\psi} =
  6(J\dot{I} - I\dot{J})/I^2= \frac{18 m^2
    v}{r^4} \frac{(1 - v^2)^{3/2} \cos\theta}{(1 -
    v^2\sin^2\theta)^{5/2} } + \mathcal{O}(r^{-4-q}).\qquad\qquad
    \label{AsymptoticPsiDot}
\end{eqnarray}
%\end{subequations}
Formula \eqref{AsymptoticPsi} is consistent with the following
expression for the Kerr spacetime
\[\Psi = -\frac{6m}{(r-ia\cos\theta)^3} = -\frac{6m}{r^3} + \mathcal{O}(r^{-4}),\]
given in terms of Boyer--Lindquist coordinates, and where $a$ denotes
the angular momentum ---see Chapter 21 of \cite{SteMacHer80}, for
example. Substitution of
\eqref{AsymptoticPsi}--\eqref{AsymptoticPsiDot} into the $3+1$
decompositions of $Q_{ab}$ ---see equation \eqref{NormalTimesQ} from
the Appendix--- then gives
%\begin{subequations}
\begin{align}
    & N^a N^b Q_{ab} = \left(\frac{9}{16m}\right)^{2/3}\frac{1}{1-v^2}
  + \mathcal{O}(r^{-q}), \label{NNQ} \\ & N^a Q_{a i} dx^i =
  \left(\frac{9}{16m}\right)^{2/3}\frac{v}{1-v^2} dz +
  \mathcal{O}(r^{-q}).\label{NQdxi}
\end{align}
%\end{subequations}
Therefore, if $(\mathcal{S}, \bm h, \bm K)$ is propagating-type-D then
equation \eqref{Saez} implies that the associated Killing vector
$\xi^a$ has lapse and shift parts $(\xi_ {N},\xi^i)$ determined by
equation \eqref{Saez2}:
%\begin{eqnarray*}
\[
\xi_N^2=N^aN^bQ_{ab}, \qquad \xi_N\xi_i dx^i=Q_{ai}N^a dx^i.
\]
%\end{eqnarray*}
Hence, from equations \eqref{NNQ}--\eqref{NQdxi} we conclude that
$\xi^a$ is a real-valued asymptotically-translational Killing vector,
given up to a possible overall sign by
\begin{equation}
\label{eq:xiAsymptoticsForSchwarz}
\xi^a = \left(\frac{3}{4m^2}\right)^{2/3}p^a + \mathcal{O}(r^{-q})
\quad \text{on}\quad \mathcal{S}.
\end{equation}
Note that we recover the result from \cite{BacVal10a, BacVal10b} that
$\xi^a \propto p^a$ at spatial infinity. In other words, $\xi^a$ is
real and asymptotes to a time translation.
\medskip

The discussion of this subsection leads to the following result, which
should be compared with Theorem 28 of \cite{BacVal10b}:
\begin{theorem}\label{Thm:AsymptoticallySchwarzschildean}
    Let $(\mathcal{S}, \bm h, \bm K)$ be a smooth boosted
    asymptotically-Schwarzschildean initial dataset 
    \eqref{AandbetaForSchwarzschildean}, of
    order $(k,q)$ where $k\geq 4$, with two asymptotically-Euclidean
    ends and satisfying
    \begin{enumerate}[label=(\roman*)]
        \item $I\neq 0$ on $\mathcal{S}$,
        \item $\psi:=-6J/I$ admits a smooth globally-defined cube root
          over $\mathcal{S}$.
    \end{enumerate}
    Then $\mathcal{I}(\mathcal{S}, \bm h, \bm K)=0$ if and only if
    $(\mathcal{S}, \bm h, \bm K)$ is locally an initial data set for
    the Kerr spacetime.
\end{theorem}
\begin{proof}
    The ``if" direction is immediate, since the Kerr spacetime is type
    D, implying that $\mathcal{I}(\mathcal{S}, \bm h, \bm K)=0$. For
    the ``only if" direction, assumption (i) and
    $\mathcal{I}(\mathcal{S}, \bm h, \bm K)=0$ imply that the local
    spacetime development is type D, by Theorem
    \ref{Coro:TypeDInvariant}. Hence, $\Psi_{ABCD}=\Psi
    o_{(A}o_{B}\iota_{C}\iota_{D)}$ for some
    $\Psi:\mathcal{S}\rightarrow\mathbb{C}$. Noting that $\Psi=\psi$
    for type D, assumption (ii) then implies that there is a
    globally-defined smooth Killing spinor,
    $\kappa_{AB}=\Psi^{-1/3}o_{(A}\iota_{B)}$, and a globally defined
    Killing vector field $\xi^a$, over $\mathcal{S}$. Now, $\xi^a$ is
    proportional to the ADM $4-$momentum at infinity ---see equation
    \eqref{eq:xiAsymptoticsForSchwarz}. Since the ADM $4$-momentum is
    timelike (see equation \eqref{ADMForSchwarz}), it follows that
    $\xi^a$ tends to a time translation as $r\rightarrow \infty$.
    Note also that the Komar mass associated to $\xi^a$ coincides with
    the ADM mass $m$ (see \cite{Ashetkar84}, for example), which is
    positive by assumption. Hence, we can apply Theorem
    \ref{Thm:JuanThomas} with $\pi_{AB}=\kappa_{AB}$ and
    $\eta^{AA'}=\xi^{AA'}$, implying that $(\mathcal{S}, \bm h, \bm
    K)$ is locally isometric to initial data for the Kerr spacetime.
\end{proof}
\begin{remark}
Note that assumption (i) is sufficient to guarantee that
$\kappa_{AB}=\Psi^{-1/3}o_{(A}\iota_{B)}$ is well-defined on any
contractible open subset $\mathcal{U}\subset\mathcal{S}$. However,
assumption (ii) is needed here to ensure that $\kappa_{AB}$, and hence
$\xi_{AA'}$, are globally-defined over $\mathcal{S}$. No such
assumption is required in \cite{BacVal10a, BacVal10b}, since the
construction is fundamentally a \emph{global} one ---the method is
based on the construction of an ``approximate Killing spinor" as the
solution to an elliptic PDE over $\mathcal{S}$.
\end{remark}

%%%%%%%%%%%%%%%%%%%%%%%%%%%%%%%%%%%%
\subsection{The more general case}
\label{Sec:more_general}
%%%%%%%%%%%%%%%%%%%%%%%%%%%%%%%%%%%

In this section, we will generalise Theorem
\ref{Thm:AsymptoticallySchwarzschildean} to a broader class of initial
data for which the property $\xi^a\propto p^a$ necessarily holds at
infinity. Our approach is similar to that of Theorem 5.1 of
\cite{AndBacBlu16}.  \\

We begin with the following two results, based on Propositions 2.2 and
3.1 of \cite{BeigChr96}:
\begin{proposition}\label{KVasympt-Lambda}
(Beig \& Chru\'{s}ciel, \cite{BeigChr96}) Let $(\mathcal{S},\bm h, \bm
  K)$ be an asymptotically Euclidean initial data set of order $(k,q)$
  with $k\geq 2$ and let $x^i$ denote asymptotically Cartesian
  coordinates. If $(N, Y^i)$ is an asymptotically KID set with $N, Y^i
  \in C^2$, then there exist constants $\Lambda_{ij}=\Lambda_{[ij]}$
  such that:
     \begin{equation}
        Y^i -\Lambda_{ij}x^i=\mathcal{O}_{k}(r^{1-q}), \quad
        N+\Lambda_{0i}x^i=\mathcal{O}_{k}(r^{1-q}).
    \end{equation}
Further,
    \begin{enumerate}[label=(\roman*)]
        \item If $\Lambda_{ij}=0$, then there exist constants $A^i$
          such that
        \begin{equation}
            Y^i-A^i =\mathcal{O}_{k}(r^{-q}), \quad N-A^0
            =\mathcal{O}_{k}(r^{-q})\label{AsymptoticallyTranslational}
        \end{equation}
        \item If $\Lambda_{ij}=A^i=A^0=0$, then $Y^i=N=0$.
    \end{enumerate}
    In case (i), we say that $(N,Y^i)$ is an \emph{asymptotically
    translational} KID set.
\end{proposition}

\begin{proposition}\label{KVAsymptotics}
(Beig \& Chru\'{s}ciel, \cite{BeigChr96}) Let $(\mathcal{S},\bm h, \bm
  K)$ be an asymptotically Euclidean initial data set of order $k\geq
  2$, $q>1/2$ and ADM $4-$momentum $p^a=(E,p^i)$ with $E>0$. Let $(N,
  Y^i)\in C^1$ be a non-trivial asymptotically translational KID set
  on $(\mathcal{S}, \bm h, \bm K)$. Then,
\[ (N, Y^i) = c(E, p^i)\]
for some constant $c\neq 0$.
\end{proposition}
Observe that Proposition \ref{KVasympt-Lambda} implies that in the
case $\Lambda_{ij} \neq 0$ the KID set $\xi^a :=(N,Y^i)$ has the
asymptotic behaviour\footnote{Here and in what follows, by ``$F\sim
r^k$" we mean that $F(r, \theta, \varphi) = f(\theta, \varphi)r^k +
o(r^k)$, for some function $f(\theta, \varphi)\not\equiv 0$, as
$r\rightarrow \infty$. } $\xi^a \sim r$, while in the $\Lambda_{ij}=0$
case one has $\xi^a \sim r^0$.  As a result, the only KID sets which
are bounded as $r\rightarrow \infty$, on initial data satisfying the
assumptions of Proposition \ref{KVasympt-Lambda}, are either trivial,
$(N, Y^i)=0$, or \emph{asymptotically translational}, case (i).  \\

We now give the main result of this section:
\begin{theorem}
\label{Thm:MoreGeneral}
    Let $(\mathcal{S}, \bm h, \bm K)$ be a smooth initial
    asymptotically-Euclidean dataset of order $(k,q)$ where $k\geq 4$,
    $q>1/2$, with two ends,
    and satisfying
    \begin{enumerate}[label=(\roman*)]
        \item $I\neq 0$ on $\mathcal{S}$,
        \item $\psi:=-6J/I$ admits a smooth globally-defined cube root
          over $\mathcal{S}$.
    \end{enumerate}
    Then $\mathcal{I}(\mathcal{S}, \bm h, \bm K)=0$ if and only if
    $(\mathcal{S}, \bm h, \bm K)$ is locally an initial data set for
    the Kerr spacetime.
    \end{theorem}
\begin{proof}
    The ``if" statement is immediate. Conversely, assumption (i) and
    $\mathcal{I}(\mathcal{S}, \bm h, \bm K)=0$ imply that the local
    spacetime development is type D, by Theorem
    \ref{Coro:TypeDInvariant}.  Hence, $\Psi_{ABCD}=\Psi
    o_{(A}o_{B}\iota_{C}\iota_{D)}$ for some
    $\Psi:\mathcal{S}\rightarrow\mathbb{C}$, and $\Psi=\psi$.
    Assumption (ii) then implies that
    $\kappa_{AB}=\Psi^{-1/3}o_{(A}\iota_{B)}$ is a well-defined,
    smooth Killing spinor, resulting in a smooth Killing vector
    $\xi^{AA'}=\nabla^{BA'}\kappa_{B}{}^A$.

    \medskip

As discussed above, the asymptotically-Euclidean conditions imply that
$\Psi_{ij}=\mathcal{O}(r^{-3})$. Note that $\Psi_{ij}$ falls off no
faster than $r^{-3}$; that is to say that it cannot be the case that
$\Psi_{ij}=o(r^{-3})$. To see this, first recall the following
expression for the ADM energy, \cite{Ashetkar84}:
\[ E = -\frac{1}{8\pi G}\lim_{r_0\rightarrow \infty}\oint_{r=r_0}rn^in^jE_{ij}~dS,\]
where $n^i$ is the unit normal to the $r=\textit{const.}$ sphere in
$\mathcal{S}$. Now suppose that $\Psi_{ij}=o(r^{-3})$, then
$E_{ij}=o(r^{-3})$ and it would follow that $E=0$. Now, by the
Positive Energy Theorem ---see Theorem 4.1 of \cite{BeigChr96}, for
example--- we conclude that $E > 0$, since equality would imply the
manifold is flat, contradicting assumption (i). Hence, we arrive at a
contradiction. Therefore, $\Psi_{ij}\sim r^{-3}$ implying that $\Psi\sim r^{-3}$,
$\kappa_{AB}\sim r$, and $\xi^{a}\sim r^0$. By Proposition
\ref{KVasympt-Lambda}, $\xi^a$ is asymptotically translational:
\[
\xi^a =\mu^a + i\nu^a \sim A^a,
\]
where $\mu^a= \Re(\xi^a)$ and $\nu^a = \Im(\xi^a)$ and $A^a$ are
complex constants.  Moreover, using Proposition \ref{KVAsymptotics} we
have that
\[
\Re{(A^a)}=c_1 p^a, \quad \Im(A^a)=c_2 p^a,
\]
where $c_1$ and $c_2$ are real constants, at least one of which is
non-zero. Now consider the Killing spinor
\[\pi_{AB} = (c_1 + ic_2)^{-1}\kappa_{AB},\]
with associated Killing vector
\[\eta^{AA'}=\nabla^{BA'}\pi_{B}{}^A = (c_1 + ic_2)^{-1}\xi^{AA'}.\]
It is clear that $\eta^a \propto p^a$ at infinity, and is therefore
asymptotically translational. Together, Theorems 4.1 and 4.2 of
\cite{BeigChr96} imply that $p^a$, and hence $\eta^a$, are timelike;
equivalently that the ADM mass is positive $m_{ADM}>0$. Also, $\eta^a$
is real-valued since its imaginary part (which is also a Killing
vector) falls off to zero at infinity and therefore is trivial by
Proposition \ref{KVasympt-Lambda}.  Again noting that the Komar mass
of $\xi^a$ coincides with $m_{ADM}$, \cite{Ashetkar84}, which is
positive by the above argument, the conclusion then follows by
application of Theorem \ref{Thm:JuanThomas}.
\end{proof}

%%%%%%%%%%%%%%%%%%%%%%%%%%%%%%%
\section{Conclusion}
%%%%%%%%%%%%%%%%%%%%%%%%%%%%%%%

We have identified simple conditions (cf. Proposition
\ref{Prop:Characterisation}) for an initial dataset $(\mathcal{S},\bm
h, \bm K)$ to give rise to a spacetime development $(\mathcal{M},\bm
g)$ that is of Petrov type D.  We call this type of data
\emph{propagating-type-D initial data} to distinguish it from initial
data which is only type D on $\mathcal{S}$, as
$\bm{\mathcal{H}}|_{\mathcal{S}}=0$ is necessary but not sufficient to
ensure the propagation of the Petrov type off $\mathcal{S}$. Using the
Killing spinor initial data equations and the Gauss constraint (the
constraint part of the Bianchi identities), it was shown that
sufficiency is obtained by requiring that certain
connection-coefficients vanish. Together, the necessary and sufficient
conditions were realised covariantly through the vanishing of a cubic
concomitant of the initial data for the Weyl spinor
$\bm{\mathcal{{H}}}|_{\mathcal{S}} $ and its time derivative
$\bm{\mathcal{\dot{H}}}|_{\mathcal{S}} $ which can be computed
directly from the initial data.  ---cf. Theorem
\ref{thm:PetrovTypeD-detect}.

\medskip

This analysis was used to define a positive semidefinite integral
curvature invariant $\mathcal{I}(\mathcal{S}, \bm h, \bm K)$, equation
\eqref{InvariantDef}, that vanishes if and only if the initial data is
propagating-type-D initial data. Hence, this invariant quantifies, at
the level of initial data, deviation from type D of the resulting
(local) spacetime development.  Finally, it was shown that, when
restricted to a class of initial data satisfying certain topological
and asymptotic conditions, the invariant vanishes if and only if the
data is locally isometric to a hypersurface of the Kerr spacetime
---cf. Theorem \ref{Thm:MoreGeneral}. This class of initial data
includes, but is not limited to, \emph{boosted
asymptotically-Schwarzschildean} data sets.  In contrast with other
notions of ``non-Kerrness" based on the Killing spinor initial data
equations, a major feature of the invariant obtained in this paper is
that it is \emph{algebraic} in the sense that its construction does
not require solving any PDE ---a solution to the ``approximate Killing
spinor" equation--- on $\mathcal{S}$ but it is rather constructed
directly from the initial data. The price to pay for this, however, is
the extra assumption that $\psi:=-6J/I$ admit a globally-defined cube
root ---see (ii) of Theorem \ref{Thm:MoreGeneral}.

\medskip

Additionally, we have provided the tensorial, as well as spinorial,
expressions for the invariant ---see equations
\eqref{eq:HInTensors}--\eqref{eq:HDotInTensors}.  That the invariant
is algebraically computable in tensorial form makes it particularly
suitable for monitoring deviations from the Kerr spacetime in the
evolution of initial data sets in Numerical Relativity. Say, for
instance in the numerical evolution of compact binaries. We also gave
an alternative invariant, equation \eqref{InvariantAlternative}, which
corresponds to the $L^2-$norms of $\bm{\mathcal{{\dot{H}}}}$ and
$\bm{\mathcal{{H}}}$.  Further work would involve studying the
evolution of the invariant under the Einstein field equations, and, on
a related note, relaxation of the regularity assumptions imposed on
the initial data.

\subsection*{Acknowledgements}
We would like to thank J. A. Valiente Kroon, D. Hilditch and
T. B\"{a}ckdahl for helpful discussions.

\appendix

\section{Normal derivative operators}\label{Ap:NormDer}

In this short appendix, we detail a calculation that allows us to
identify the tensor equivalent of the operators $D_N$ and
$\mathcal{P}$, as given in section \ref{Sec:NotationAndConventions}.
First notice that, from the definition of $D_N$ in equation
\eqref{eq:DNSpinor}, for a symmetric valence-2 spinor $\nu_{AB}$ one
has
\begin{equation}
    D_{N}\nu_{AB}=\mathcal{P}\nu_{AB} -
    A_{(A}{}^C\nu_{B)C}.\label{eq:DNvalence2Spinor}
\end{equation}
Using equation \eqref{3dimVolume} a short calculation shows that
\[A_{(A}{}^{C}\nu_{B)C}= \frac{i\sqrt{2}}{2} \epsilon_{ABCDEF}A^{CD}\nu^{EF}.\] Using this fact and that $a_i=-\frac{1}{\sqrt{2}}\sigma_{i}{}^{AB}A_{AB}$, where $\sigma_{i}{}^{AB}$ are the spatial Infeld-van-der Waerden symbols,
one obtains equation \eqref{eq:PtoDNVectors} as the tensorial
counterpart of equation \eqref{eq:DNvalence2Spinor}.

\medskip

Now, let $\nu_a$ satisfy $N^a\nu_a=0$, and
$\varphi:\mathcal{S}\hookrightarrow\mathcal{M}$ so that $h_{i}{}^a$
denotes the projector: $\varphi^*(\nu)_{i}=h_{i}{}^a\nu_a$.  Then, in
space-spinors, the projected normal derivative of a covector reads
\begin{align*}
    &\sigma{}^{i}_{AB} \; h_i{}^bN^c\nabla_c \nu_b\\ &=
  \sqrt{2}N_{(A}{}^{A'}\mathcal{P}\nu_{B)A'} \\ &=
  2N_{(A}{}^{A'}\mathcal{P}(N^D{}_{\vert A'\vert}\nu_{B)D})\\ &=
  2N_{(A}{}^{A'}N^D{}_{\vert A'\vert }\mathcal{P}\nu_{B)C} +
  2\nu_{C(A}N_{B)}{}^{A'}\mathcal{P}N^C{}_{A'} \\ &=
  \mathcal{P}\nu_{AB} - A_{(A}{}^C\nu_{B)C} \\ &= D_{N}\nu_{AB}.
\end{align*}
Translating into tensors, we then arrive at equation
\eqref{eq:DNForTensors}.

\section{Asymptotic expansions of the Weyl tensor}
\label{Sec:Asymptotics}

Here we give some asymptotic expansions for the Weyl tensor and its
derivatives, relevant for section \ref{Sec:boostedAS}.

\medskip

Using the Gauss-Codazzi equations \eqref{GCM1}--\eqref{GCM2} one can
express the electric and magnetic parts of the initial data for the
Weyl curvature; a long but direct calculation shows that for data of
the form \eqref{HuangSolh}--\eqref{HuangSolK} one has:
\[ E_{ij} = \mathcal{E}_{ij} + \mathcal{O}_1(r^{-3-q}), \quad B_{ij} = \mathcal{B}_{ij} + \mathcal{O}_1(r^{-3-q}),\]
where
%\begin{subequations}
\begin{align}
    & \mathcal{E}_{ij}dx^i dx^j \nonumber \\ &\quad =
  -\tfrac{1}{2r^3}(4A + \partial^2_\theta\alpha +\cot\theta
  \partial_\theta\alpha+\csc^2\theta ~\partial^2_{\varphi}\alpha +
  2\alpha) dr^2 \nonumber \\ &\quad +\tfrac{1}{2r}(2A + \alpha +
  \cot\theta
  \partial_{\theta}\alpha+\csc^2\theta~\partial_{\varphi}^2\alpha)d\theta^2\nonumber
  \\ &\quad +\tfrac{1}{2r}\sin^2\theta (2A + \alpha +
  \partial_{\theta}^2\alpha)d\varphi^2 \nonumber \\ &\quad -
  \tfrac{1}{r}(\partial_{\theta} - \cot\theta)\partial_{\varphi}\alpha
  d\theta d\varphi, \label{AsympE}\\ & \mathcal{B}_{ij}dx^i dx^j
  \nonumber \\ &\quad = \tfrac{2}{r^2}\csc\theta~\partial_{\varphi}
  \beta dr d\theta -\tfrac{2}{r^2}\sin\theta~\partial_{\theta} \beta
  dr d\varphi.\label{AsympB}
\end{align}
%\end{subequations}
in terms of the standard spherical coordinates $(r, \theta, \varphi)$,
related to $(x_1,x_2,x_3)$ by
\[ x_1=r \cos\theta \sin\varphi, \quad x_2=r \sin\theta \sin\varphi, \\
\quad x_3 = r\cos\varphi.\] Additionally, it is easily shown that
%\begin{subequations}
    \begin{align}
\underline{\underline{\text{rot}_2(\mathcal{B})}} &=
\left(\begin{array}{ccc} * & -\tfrac{2}{r\sin\theta}
  \mathcal{B}_{r\varphi} &
  \tfrac{2\sin\theta}{r}\mathcal{B}_{r\theta}\\ -\tfrac{2}{r\sin\theta}
  \mathcal{B}_{r\varphi} & * &
  *\\ \tfrac{2\sin\theta}{r}\mathcal{B}_{r\theta} & * & *
    \end{array} \right), \label{eq:RotB}\\
    \underline{\underline{\text{rot}_2(\mathcal{E})}} &=
    \nonumber\\ &\hspace{-10mm}\left(\begin{array}{ccc} 0 & * & * \\ *
      & \tfrac{2}{r\sin\theta}\mathcal{E}_{\theta\varphi} &
      \tfrac{1}{r\sin\theta}\mathcal{E}_{\varphi\varphi} -
      \tfrac{\sin\theta}{r}\mathcal{E}_{\theta\theta} \\ * &
      \tfrac{1}{r\sin\theta}\mathcal{E}_{\varphi\varphi} -
      \tfrac{\sin\theta}{r}\mathcal{E}_{\theta\theta} &
      -\tfrac{2\sin\theta}{r}\mathcal{E}_{\theta\varphi}
    \end{array} \right),   \label{eq:RotE}
\end{align}
%\end{subequations}
in the $dr, d\theta, d\varphi$ co-basis, where the entries denoted $*$
are omitted as they not needed for the purposes of this paper.

\section{The canonical Killing vector}\label{Sec:Killingvectorexpression}
In this appendix, we show that the Killing vector in type D spacetimes
singled out by Ferrando-Saez in \cite{FerSae09}, in fact coincides
with the canonical Killing vector
$\xi^{AA'}:=\nabla^{BA'}\kappa_{B}{}^A$ where $\kappa_{AB}$ is the
Killing spinor.  To do so, one starts by substituting
$\Psi_{ABCD}=\Psi o_{(A}o_{B}\iota_{C}\iota_{D)}$ into the second
Bianchi identity \eqref{Bianchi} which gives
%\begin{subequations} 
    \begin{flalign}
        & o^{A} o^{B} \bar{o}^{A'} \nabla_{BA'}o_{A} =
      0, \label{BianchiFirst}\\ & \iota^{A} \iota^{B} \bar{o}^{A'}
      \nabla_{BA'}\iota_{A} = 0, \\ & \bar{\iota}^{A'} o^{A} o^{B}
      \nabla_{BA'}o_{A} = 0, \\ & \iota^{A} \iota^{B} \bar{\iota}^{A'}
      \nabla_{BA'}\iota_{A} = 0,\\ & \iota^{A} o^{B} \bar{o}^{A'}
      \nabla_{AA'}o_{B} = \frac{o^{A} \bar{o}^{A'} \nabla_{AA'}\Psi}{3
        \Psi}, \\ & \iota^{A} o^{B} \bar{o}^{A'} \nabla_{BA'}\iota_{A}
      = - \frac{\iota^{A} \bar{o}^{A'} \nabla_{AA'}\Psi}{3 \Psi}, \\ &
      \iota^{A} \bar{\iota}^{A'} o^{B} \nabla_{AA'}o_{B} =
      \frac{\bar{\iota}^{A'} o^{A} \nabla_{AA'}\Psi}{3 \Psi}, \\ &
      \iota^{A} \bar{\iota}^{A'} o^{B} \nabla_{BA'}\iota_{A} = -
      \frac{\iota^{A} \bar{\iota}^{A'} \nabla_{AA'}\Psi}{3
        \Psi}. \label{BianchiLast}
    \end{flalign}
%\end{subequations}
Defining $\mathcal{U}_{AA'BB'} = \epsilon_{A'B'} o_{(A}\iota_{B)}$, it
is straightforward to show by expanding in spin dyad components that
\[\mathcal{Q}_{AA'BB'CC'DD'}=\Psi \mathcal{U}_{AA'BB'}\mathcal{U}_{CC'DD'}.\]
Let $\xi_{AA'} = \nabla^B{}_{A'}\kappa_{AB}$. Substituting
$\kappa_{AB}=\Psi^{-1/3}o_{(A}\iota_{B)}$, along with the identities
\eqref{BianchiFirst}--\eqref{BianchiLast}, gives
\begin{align*}
     o^{A} \bar{o}^{A'}\xi_{AA'} &= - \frac{o^{A} \bar{o}^{A'}
       \nabla_{AA'}\Psi}{2
       \Psi^{4/3}}\\ &=\tfrac{3}{2}\Psi^{-1/3}o^{A}\bar{o}^{A'}(\nabla^{BB'}\mathcal{U}_{BB'AA'}),
     \\ o^A\bar{\iota}^{A'} \xi_{AA'} &= -\frac{\bar{\iota}^{A'} o^{A}
       \nabla_{AA'}\Psi}{2
       \Psi^{4/3}}\\ &=\tfrac{3}{2}\Psi^{-1/3}o^{A}\bar{\iota}^{A'}(\nabla^{BB'}\mathcal{U}_{BB'AA'}),\\ \iota^{A}\bar{o}^{A'}
     \xi_{AA'} &= \frac{\iota^{A} \bar{o}^{A'} \nabla_{AA'}\Psi}{2
       \Psi^{4/3}}\\ &=\tfrac{3}{2}\Psi^{-1/3}\iota^{A}\bar{o}^{A'}(\nabla^{BB'}\mathcal{U}_{BB'AA'}),\\ \iota^{A}
     \bar{\iota}^{A'} \xi_{AA'} &= \frac{\iota^{A} \bar{\iota}^{A'}
       \nabla_{AA'}\Psi}{2
       \Psi^{4/3}}\\ &=\tfrac{3}{2}\Psi^{-1/3}\iota^{A}\bar{\iota}^{A'}(\nabla^{BB'}\mathcal{U}_{BB'AA'}).
\end{align*}
It follows that
\[ \xi_{AA'} = \tfrac{3}{2}\Psi^{-1/3}\nabla^{BB'}\mathcal{U}_{BB'AA'}.\]
Contracting equation \eqref{Qdef} with $N^c$ and performing a $3+1$
decomposition, we get
\begin{align}
    N^cQ_{ca} = & \frac{\psi\dot{\psi} D_{a}\psi - 6
    \dot{\psi} \psi_{ac}D^{c}\psi + 6i N^b
    \epsilon_{badf}\psi_{c}{}^{f} (D^{d}\psi)(D^{c}\psi)}{12
    \psi^{11/3}}\nonumber \\ &+ \frac{(- \psi (D^b\psi)(D_b\psi) + 6
    \psi_{bc}(D^{b}\psi)(D^{c}\psi))}{12
    \psi^{11/3}}N_a,\label{NormalTimesQ}
\end{align}
where
\begin{align} 
\dot{\psi} &:= N^a\nabla_a\psi = N^a\nabla_a(-6J/I) = 6I^{-2}(J\dot{I}
- I\dot{J}).
\end{align}

\section{Wave equations in boosted regions}\label{Sec:Boost}

\noindent In this Appendix we collect and adapt some of the results of \cite{ChrMur81} that are used in main text of this article.
Let $U \subset \mathbb{R}^4$ be an open set and define $\sigma(x)$ as
\[
\sigma(x):=(1+|x|^2)^{1/2}
\] 
where $x\in \mathbb{R}^4$. With this notation at hand, one defines a weighted Sobolev space $H_{s,\delta}(U)$ with $s \in \mathbb{N}$, $\delta \in \mathbb{R}$ as the Hilbert space of vector valued functions with inner product:
\begin{align*}
    \langle \bm u_1, \bm u_2 \rangle_{H_{s,d}(U)}:=\sum_{|\alpha| \leq s}\langle \sigma^{\delta +|\alpha|}D^\alpha \bm u_1, \sigma^{\delta +|\alpha|}D^\alpha \bm u_2 \rangle_{L^2(U)},
\end{align*}
where $\langle \cdot, \cdot \rangle_{L^{2}}(U)$ denotes the standard $L^2$-inner product on $U$. The associated norm is naturally given by
\begin{align}\label{def_Weighted_Sobolev}
    || \bm u ||^2_{H_{s,\delta}(U)}:=\langle \bm u, \bm u \rangle_{H_{s,\delta}(U)}.
\end{align}
One defines a \emph{boost-region} as
\begin{align}
    \Omega_{\vartheta} := \{ x \in \mathbb{R}^4 \;\; |\;\; \frac{|t|}{\bar{\sigma}(\bar{x})}<\vartheta\},
\end{align}
 where $x=(x^0,\bar{x})=(t,x^i)$ are asymptotically Cartesian coordinates, and one defines as usual $r^2 = \delta_{ij}x^ix^j$ so that
  $\sigma(\bar{x}):=(1+r^2)^{1/2}$ and $0<\vartheta<1$ is a constant.
  \\
 The function
  \[
  \tau(x)=\frac{t}{\bar{\sigma}(\bar{x})}
  \]
  induces a foliation on $\Omega_{\vartheta}$ given by
  \[\Omega_{\vartheta}= \bigcup_{\tau \in (-\vartheta,\vartheta)}\mathcal{S}_{\tau}\]
  By convention in this article $\mathcal{S}_0=\mathcal{S}$.
  The relation between the Euclidean measures on $(-\vartheta,\vartheta)\times \mathbb{R}^3$ and $\Omega_{\vartheta}$, $dx=\bar{\sigma}(\bar{x})d\tau d\bar{x}$ gives a relation between the weighted Sobolev spaces on $\Omega_{\vartheta}$ and on the leaves of the foliation $\mathcal{S}_{\tau}$  ---see \cite{ChrMur81} for a detailed discussion.
    To simplify the notation, the boost region $\Omega_{\vartheta}$ will be denoted by $\Omega$.
\\
\noindent Now, consider second order differential operators:
\begin{align}
L \bm u=\sum_{\alpha=0} ^{2} \bm a_\alpha \cdot D^\alpha \bm u
\end{align}
where $\bm u$ and $L \bm u$ are $\mathbb{R}^4$-vector valued functions in $\Omega$.
$L$ will be said to satisfy the weak coupling and hyperbolicity assumption (hypothesis I in \cite{ChrMur81}) if
\begin{align}\label{hypothesisI}
    \bm a_2 = \bm g^\sharp\cdot Id 
\end{align}
where $\bm g$ is a regular Lorentzian metric in $\Omega$ and $Id$ is the identity. 
Additionally, we will say that
$L$ satisfies the regularity assumption (hypothesis II in \cite{ChrMur81}) if for $\alpha=0,1,2$ there are non-negative integers and real numbers $(s_\alpha, \delta_\alpha)$ such that  \begin{equation}\label{hypothesisII}
    \left\{
\begin{array}{ll} 
         \bm a_\alpha \in H_{s_\alpha,\delta_\alpha} (\Omega) \qquad \text{for}\;\; \alpha=0,1 \\
         \bm g-\bm \eta \in H_{s_2,\delta_2}(\Omega)
\end{array} \right.
\end{equation}
where 
 \begin{align}\label{s-delta-conds}
 s_\alpha> 1+\alpha, \quad \delta_\alpha> -\alpha . \end{align}
  For later use denote $s':=\min_{0 \leq \alpha\leq 2}\{s_\alpha\}+1$

  \medskip
  
With these definitions, Theorems 4.1 and 4.2 of \cite{ChrMur81}
combined can be phrased as follows:
 
\begin{proposition}(Christodoulou \& O'Murchadha, \cite{ChrMur81})\label{prop:ChristOMurch}
Let $L$ be a second-order operator satisfying
\eqref{hypothesisI}--\eqref{hypothesisII}. Let $\mathcal{S}$ be the
initial hypersurface with normal vector $\bm N$ and let $\Omega
\subset {D}^{+}(\mathcal{\mathcal{S}}) $ be a boost-type domain. If
\begin{align}
    \bm \beta \in H_{\lambda-1,\delta+2}(\Omega) \\ \bm \phi \in
    H_{\lambda, \delta +1/2} (\mathcal{S}), \\ \bm \psi \in
    H_{\lambda-1,\delta +3/2}(\mathcal{S})
\end{align}
with $2 \leq \lambda \leq s'$ and $\delta \in \mathbb{R}$, then the
Cauchy problem
       \begin{equation}
    \left\{
\begin{array}{ll}
	 L \bm u = \bm \beta & \qquad
         \text{on}~\mathcal{D}^{+}(\mathcal{\mathcal{S}}),\\ \bm u=\bm
         \phi, & \qquad \text{on}~\mathcal{\mathcal{S}},\\ \nabla_N
         \bm u= \bm \psi & \qquad\text{on}~\mathcal{\mathcal{S}}.
\end{array} \right.
\end{equation}
 has a unique solution
\begin{equation}
    \bm u \in H_{2,\delta}(\Omega)
\end{equation}
\end{proposition}

Similarly, the boost theorem (Theorem 6.1 in \cite{ChrMur81}) can be
phrased minimally as follows:
\begin{theorem}(Christodoulou \& O'Murchadha)\label{Thm:boost_theorem}
    Let $(\bm h, \bm K, \mathcal{S})$ be an initial data set for the
    Einstein field equations, with
    \begin{align}
        \bm h - \bm e \in H_{s, \delta+1/2}(\mathcal{S}), \qquad \bm K
        \in H_{s-1, \delta + 3/2}(\mathcal{S})
    \end{align}
    where $\bm e$ is the Euclidean metric and $s\geq 4$ and $\delta
    >-2$. Then, there exists a solution to the Einstein field
    equations $(\bm g, \Omega)$ in some boost type domain $\Omega$
    such that
    \begin{align}
        \bm g - \bm \eta \in H_{s,\delta}(\Omega)
    \end{align}
    with $(\bm h, \bm K)$ being the first and second fundamental forms
    of $\mathcal{S}$ relative to $\bm g$.
\end{theorem}
\medskip

We want to use the above results in the simpler set-up of this
article. Namely, for functions which are smooth and have certain decay
at infinity.  Using the fact that the volume-form in Euclidean space
$(\bm e, \mathbb{R}^3)$ expressed in spherical polar coordinates is
given by $dV = r^2 \sin^2 \theta dr \wedge d \theta \wedge d\phi$,
direct inspection of definitions \eqref{def_Weighted_Sobolev} shows
that:
\begin{align}\label{eq_keyInclusion}
    C^\infty(\mathcal{S})\cap \mathcal{O}_{s}(r^{-\delta'-3/2})\subset
    H_{s,\tilde{\delta}}(\mathcal{S}).
\end{align}
for any $\tilde{\delta} < \delta'$.

\medskip

Now, if we assume $(\bm h, \bm K, \mathcal{S})$ is asymptotically Euclidean of order $(k, q)$. Using that we are working in smooth category, along with  the inclusion \eqref{eq_keyInclusion} gives
\[\bm h -\bm e \in H_{k,\tilde{q}-3/2}(\mathcal{S}), \qquad
\bm K \in H_{k-1,\tilde{q}-1/2}(\mathcal{S}).\]
for any $\tilde{q}<q$.
Relabelling, to more easily apply the boost theorem, one takes
$\delta <  q-2$ and concludes, using Theorem \ref{Thm:boost_theorem}, that 
\begin{eqnarray}\label{metric_and_ders_decayBoostdomain}
    \bm g - \bm \eta\in H_{k, \delta}(\Omega) \\ \bm\Gamma \in H_{k-1,\delta+1} (\Omega) \\
    \bm\partial\bm\Gamma \in H_{k-2,\delta+2}(\Omega) .
\end{eqnarray}
As discussed in Remark \ref{remark:kid_translational}, the Killing vector candidate satisfies the IVP
\begin{equation}\label{CauchyProbForX}
    \left\{
\begin{array}{ll}
	 \square X^a = 0 & \qquad
         \text{on}~\mathcal{D}^{+}(\mathcal{U}),\\ 
         X^a =(N, Y^i)
         &\qquad\text{on}~\mathcal{U},\\ 
         N^b\nabla_b X^a = -D^a N + K_{b}{}^aY^b
         &\qquad\text{on}~\mathcal{U}.
\end{array} \right.
\end{equation}
where $\mathcal{U}\subset\mathcal{S}$ is an open set.
Let $Z^a = X^a-A^a$ where $A^a$ are constants and
\begin{eqnarray*}
    Z^a \in \mathcal{O}_k(r^{-q}), \qquad \nabla_N Z^a \in \mathcal{O}_{k-1}(r^{-1-q})
  \quad  \text{on}\quad  \mathcal{U}
\end{eqnarray*}
Recalling we are working in the smooth category and using again the inclusion
\eqref{eq_keyInclusion} gives
\begin{eqnarray*}
    Z^a \in H_{k,\delta+1/2}(\mathcal{U}), \quad \nabla_N Z^a \in H_{k-1,\delta +3/2}(\mathcal{U})
\end{eqnarray*}
for $\delta < q-2$. Additionally, the following calculation
\begin{align*}
    \square Z^c &=-\square A^c \\
    &= -g^{ab}\partial_{a}\partial_{b}A^c + [\bm g^{\sharp}\bm \cdot \bm \Gamma \cdot \partial A]^c + [\bm g ^\sharp \cdot (\bm\Gamma^2 +\partial \bm\Gamma) \cdot A]^c\\
    &= [\bm g ^\sharp \cdot (\bm\Gamma^2 +\partial \bm\Gamma) \cdot A]^c 
    %\\
   % & \in H_{k-2, \delta + 2}(\Omega), 
\end{align*}
shows that $\square Z^c =-\square A^c \in H_{k-2,\delta +2}( \Omega)$. Rewriting these results with $\lambda = k-1$ gives:
\begin{eqnarray}
\square Z^c \in H_{\lambda -1,\delta +2}(\Omega)\\
Z^a \in H_{\lambda+1,\delta+1/2}(\mathcal{S})\subset H_{\lambda,\delta+1/2}(\mathcal{S})\\
\nabla_N Z^a \in H_{\lambda,\delta+3/2}(\mathcal{S})\subset H_{\lambda-1,\delta+3/2}(\mathcal{S}),
\end{eqnarray}
where the expressions have been arranged so that they match the assumptions of Proposition
\ref{prop:ChristOMurch}. Then, to apply Proposition \ref{prop:ChristOMurch} it only remains to verify that conditions \eqref{hypothesisI} and \eqref{hypothesisII} hold. To do so, we simply notice that as long as $\bm g$ 
is Lorentzian, condition \eqref{hypothesisI} holds and condition \eqref{hypothesisII} is verified from equations \label{metric_and_ders_decayBoostdomain}
by taking $k>3$ and $\delta>-2$. This ensures that IVP implied by  \eqref{CauchyProbForX} has unique solution 
\begin{equation}
    Z^a \in H_{2,\delta}(\Omega).
\end{equation}
 This means, in terms of the original labels $(k,q)$  that we need $k>3$ and $q>0$. Finally, as a consequence of the Sobolev embedding ---see Theorem 2.1 of \cite{ChrMur81}, for example--- we have the following inclusion
\begin{align}\label{otherInclusion}
    H_{s,\delta}(\Omega) \subset \mathcal{O}_{s-2}(r^{-\delta-3/2}).
\end{align}
In particular, taking $s=2$, we see that $Z^a \rightarrow 0$ as $r\rightarrow\infty$, on $\Omega$.

\bibliographystyle{sn-mathphys-num}

\end{document}